\documentclass[sigconf]{acmart}

\AtBeginDocument{%
  }

\setcopyright{acmcopyright}
\copyrightyear{2024}
\acmYear{2024}
\acmDOI{XXXXXXX.XXXXXXX}

\acmConference[ISSAC '24]{ISSAC 2024}{July 16--19,
  2024}{Raleigh, NC, USA}
\acmPrice{15.00}
\acmISBN{978-1-4503-XXXX-X/18/06}

\usepackage{enumitem}
\usepackage{bbm,tikz}
\usepackage{algorithm}
\usepackage{algpseudocode}
\algrenewcommand\algorithmicrequire{\textbf{Input:}}
\algrenewcommand\algorithmicensure{\textbf{Output:}}
\usepackage{amsmath,bm}
\usepackage{parallel,enumitem}
\usepackage{todonotes}
\usepackage{xspace,setspace}

\newcommand{\V}{\bm{V}}

\newcommand{\R}{\mathbb{R}}

\newcommand{\N}{\mathbb{N}}
\newcommand{\Q}{\mathbb{Q}}

\newcommand{\xb}{\mathbf{x}}
\newcommand{\yb}{\mathbf{y}}
\newcommand{\ab}{\mathbf{a}}
\newcommand{\bb}{\mathbf{b}}
\newcommand{\gb}{\bm{g}}
\newcommand{\C}{\mathbb{C}}

\newcommand{\Cn}{\C^n}

\def\1{{\mathbbm 1}}
\newcommand{\gbt}{\widetilde{\gb}}

\newcommand{\remi}[1]{\textcolor{purple}{#1}}
\newcommand{\todoremi}[2][]{\todo[color=pink,#1]{#2}}
\newcommand{\todoremigh}[3][]{\todo[color=pink,#1]{#2}{\remi{#3}}}

\newcommand{\compose}{\textsf{Compose}\xspace}

\newcommand{\InRadical}{\textsf{InRadical}\xspace}
\newcommand{\VectorBasis}{\textsf{VectorBasis}\xspace}
\newcommand{\InvariantSet}{\textsf{InvariantSet}\xspace}
\newcommand{\CheckPI}{\textsf{CheckPI}\xspace}

\theoremstyle{plain}
\newtheorem{theorem}{Theorem}[section]
\newtheorem{lemma}[theorem]{Lemma}
\newtheorem{proposition}[theorem]{Proposition}

\newtheorem{definition}[theorem]{Definition}
\theoremstyle{definition}
\newtheorem{example}{Example}
\theoremstyle{remark}
\newtheorem{remark}{Remark}

\newcommand{\programbox}[2][\linewidth]{
\begin{samepage}\normalfont
\vspace*{.5em}\hspace*{0.3cm}\fbox{
\hspace*{-0.3cm}\begin{minipage}{#1}
\vspace*{-.1em}\begin{algorithmic}\setstretch{1}
#2
\end{algorithmic}
\vspace*{-.2em}
\end{minipage}
}\\
\end{samepage}
}

\newcommand{\programboxappendix}[2][\linewidth]{
\begin{samepage}\normalfont
\vspace*{.5em}\hspace*{0.0cm}\fbox{
\hspace*{-0.3cm}\begin{minipage}{#1}
\vspace*{-.1em}\begin{algorithmic}\setstretch{1}
#2
\end{algorithmic}
\vspace*{-.2em}
\end{minipage}
}\\
\end{samepage}
}

\begin{document}

\title{Algebraic Tools for Computing Polynomial Loop Invariants %in the Case of Unsolvable While Loops
}

\author{Erdenebayar Bayarmagnai}
\email{erdenebayar.bayarmagnai@kuleuven.be}
%\email{webmaster@marysville-ohio.com}
\affiliation{%
  \institution{KU Leuven}
 % \streetaddress{P.O. Box 1212}
  \country{Belgium}
}

\author{Fatemeh Mohammadi}
\email{fatemeh.mohammadi@kuleuven.be}
\affiliation{%
  % \institution{UiT The Arctic University of Norway}
 \institution{KU Leuven}
 % \streetaddress{1 Th{\o}rv{\"a}ld Circle}
 % \city{Hekla}
  \country{Belgium}}

\author{R\'emi Pr\'ebet}
\email{remi.prebet@kuleuven.be}
\affiliation{%
  % \institution{UiT The Arctic University of Norway}
 \institution{KU Leuven}
 % \streetaddress{1 Th{\o}rv{\"a}ld Circle}
 % \city{Hekla}
  \country{Belgium}}

\renewcommand{\shortauthors}{Bayarmagnai, Mohammadi and Pr\'ebet}

\begin{abstract}
Loop invariants are properties of a program loop that hold before and after each iteration of the loop. They are often employed to verify programs and ensure that algorithms consistently produce correct results during execution. Consequently, the generation of invariants becomes a crucial task for loops. We specifically focus on polynomial loops, where both the loop conditions and assignments within the loop are expressed as polynomials. Although computing polynomial invariants for general loops is undecidable, efficient algorithms have been developed for certain classes of loops. For instance, when all assignments within a while loop involve linear polynomials, the loop becomes solvable. In this work, we study the more general case where the polynomials exhibit arbitrary degrees.

Applying tools from algebraic geometry, we present two algorithms designed to generate all polynomial invariants for a while loop, up to a specified degree. These algorithms differ based on whether the initial values of the loop variables are given or treated as parameters. Furthermore, we introduce various methods to address cases where the algebraic problem exceeds the computational capabilities of our methods. In such instances, we identify alternative approaches to generate specific polynomial invariants.

%Loop invariants are properties of a program loop that are true both before and after each iteration of the loop. They are often employed to verify programs and ensure that algorithms consistently produce correct results before execution. %\todoremigh{soften the statement}{One of the most fundamental and useful types of invariants is the polynomial invariant}, which represents a polynomial equation involving program variables. 
%Consequently, generating invariants becomes a crucial task for loops. We introduce the concept of invariant varieties and establish that they are sufficient to encompass all polynomial invariants for a while loop, up to a given degree. Our methods are rooted in tools from computational algebraic geometry. We present multiple algorithms designed to generate all polynomial invariants for a while loop. These algorithms differ based on whether the initial values of the loop variables are predetermined or treated as variables. Furthermore, we introduce various methods to handle cases where the algebraic problem exceeds the capabilities of our computational methods. In such instances, we identify alternative approaches to generate specific polynomial invariants.
\end{abstract}
\thanks{%
The authors are partially supported by the KU Leuven grant iBOF/23/064, the FWO grants G0F5921N and G023721N, and the UiT Aurora project MASCOT.
  }



 % \includegraphics[width=\textwidth]{sampleteaser}
 % \caption{Seattle Mariners at Spring Training, 2010.}
 % \Description{Enjoying the baseball game from the third-base
 % seats. Ichiro Suzuki preparing to bat.}
  %\label{fig:teaser}

\maketitle

\newpage
%%%%%%%%%%%%%%%%%%%%%%%%%%%%%%%%%%%%%%%%%%%%%%%%%%%%%%%%%%%%%%%%%%%%%%%%
\section{Introduction}
%%%%%%%%%%%%%%%%%%%%%%%%%%%%%%%%%%%%%%%%%%%%%%%%%%%%%%%%%%%%%%%%%%%%%%%%

Loop invariants denote properties that hold both before and after each iteration of a loop within a given program. They play a crucial role in automating the program verification, ensuring that algorithms consistently yield correct results prior to execution.
Notably, 
various recognized methods for safety verification like the Floyd–Hoare inductive assertion technique \cite{floyd1993assigning} and the termination verification via standard ranking functions technique \cite{manna2012temporal} rely on loop invariants to verify correctness, ensuring complete automation in the verification process. 
%In this study, our emphasis lies on polynomial imperative loops featuring complex or real variables, where all assignments within the loop consist of polynomials over program variables. 

In this work, we focus on polynomial loops, wherein expressions within assignments and conditions are polynomials equations in program variables. 
%Polynomial invariants, characterized by polynomial equations involving program variables, 
%\todoremigh{why?}{stand out as invaluable tools in program analysis}. 
More precisely, a polynomial loop is of the form:
% \begin{samepage}

\programbox[0.55\linewidth]{
\State$(x_1, x_2,\ldots, x_n)=(a_1,a_2,\ldots,a_n)$
\While{$g_1 = \cdots = g_k =0$}
\State $\begin{pmatrix}
x_1 \\
x_2 \\
\vdots \\
x_n
\end{pmatrix}
\xleftarrow{\textbf{F}}
\begin{pmatrix}
f_1\\
f_2\\
\vdots\\
f_n
\end{pmatrix}
$
\EndWhile
}\label{page:alg}

\noindent where the $x_i$'s represent program variables with initial values $a_i$, and $g_i$'s and $f_i$'s are polynomials in the program variables.  %\textcolor{blue}{We later focus on the case when the points satisfying the conditions on the while guard are corresponding to a semi-algebraic set, that is solutions of systems of polynomial inequalities.} 
%For more details, see Definition~\ref{def:loop}.
Computing polynomial invariants for loops has been a subject of study over the past two decades, see e.g.~\cite{de2017synthesizing, %ISSAC2023,
kovacs2023algebra, Unsolvableloops, hrushovski2018polynomial, karr1976affine,kovacs2008reasoning,rodriguez2004automatic,rodriguez2007automatic, rodriguez2007generating}.
Computing polynomial invariants for general loops is undecidable \cite{hrushovski2023strongest}. Therefore, particular emphasis has been placed on specific families of loops, especially those in which the assertions are all linear or %the loop is solvable—that is, it 
can be reduced to linear assertions. In the realm of linear invariants, Michael Karr introduced an algorithm pioneering the computation of all linear invariants for loops where each assignment within the loop is a linear function \cite{karr1976affine}. Subsequent studies, such as \cite{muller2004note} and \cite{rodriguez2007automatic}, have demonstrated the feasibility of computing all polynomial invariants up to a specified degree for loops featuring linear assignments. Further, the problem of generating {\em all} polynomial invariants for loops with linear assignments, %and solvable loops 
are studied in \cite{hrushovski2018polynomial} and \cite{rodriguez2007automatic}. 

Another class of loops for which invariants have been successfully computed is the family of solvable loops. These loops are characterized by polynomial assignments that are either inherently linear or can be transformed into linear forms through a change of variables, as elaborated in \cite{de2016polynomial} and \cite{kovacs2008reasoning}. Nonetheless, challenges persist when dealing with loops featuring non-linear or unsolvable assignments, as discussed in \cite{rodriguez2007generating} and \cite{Unsolvableloops}.

Before stating our main results, we introduce some terminology from algebraic geometry. We refer to \cite{cox2013ideals,kempf_1993} for further details. 
Let $\mathbb{C}$ denote the field of complex numbers.
Let $S$ be a set of polynomials in $\mathbb{C}[x_1, \ldots, x_n]$, then the \emph{algebraic variety} $\V(S)$ associated to $S$ is the common zero set of all polynomials in $S$. 
%$$V(S)=\{x\in \mathbb{R}^n\mid f(x)=0 \text{ for all } f\in S\}.$$ 
%Note that, contrary to some authors in the literature, we do not assume here that an algebraic variety is irreducible.
Here, $\V(S) = \V(\langle S \rangle)$, where $\langle S \rangle$ is the ideal generated by~$S$.
Conversely, the \emph{defining ideal} of a subset $X\subset\mathbb{C}^n$ is the set of polynomials in $\mathbb{C}[x_1, \ldots, x_n]$ that vanish on $X$. %; this is an ideal of $\mathbb{C}[x_1, \ldots, x_n]$.
%\[
%I(X) = \{ f \in \mathbb{C}[x_1, \ldots, x_n] \mid f(x) = 0 \text{ for every } x \in X \}.
%\]
The algebraic variety associated to the ideal $I(X)$ is called the \emph{Zariski closure} of $X$. Therefore, if $X$ is an algebraic variety, then $\V(I(X)) = X$. Moreover, $X_1\subseteq X_2$ implies that $I(X_2) \subseteq I(X_1)$. 
% According to \cite[Page 232]{cox2013ideals}, \todo{Do we need this?} $I(X)$ has a (unique) decomposition  into primary ideals $I(X) = Q_1 \cap \cdots \cap Q_n$. Moreover, $X = \V(Q_1) \cup \cdots \cup \V(Q_n)$ is an irreducible decomposition of $X$.
A map $F: \C^n \to \C^m$ is called a \emph{polynomial map}, if there exist $f_1,\ldots,f_m$ in $\C[x_1,\ldots,x_n]$, such that $F(x) = (f_1(x),\ldots,f_m(x))$ for all $x \in \C^n$. For the sake of simplicity, in what follows we will refer to polynomial maps and their associated polynomials interchangeably.

\smallskip
We now define the main object introduced in this paper.
\begin{definition}\label{def:invariantset}
Let $F : \Cn \longrightarrow \Cn$ be a polynomial map and $X\subseteq\Cn$ an algebraic variety.
The invariant set of $(F,X)$ is defined as:
\begin{center}
    $S_{(F,X)} = \{x \in X \mid \forall m\in \N, F^{(m)}(x) \in X\},$
\end{center}
where $F^{(0)}(x)=x$ and $F^{(m)}(x)= F(F^{(m-1)}(x))$ for any $m>1$.
\end{definition}
\noindent\textbf{Our contributions.} In this work, we consider the problem of generating polynomial invariants for loops with polynomial maps of arbitrary degrees. Our contributions are as follows:

\hspace*{-0.9cm}
\parbox{1.05\linewidth}{
\begin{enumerate}[label=(\arabic*)]
\item We 
% first show that invariant sets are algebraic varieties and 
design Algorithm~\ref{algo1} to compute invariant sets and use it
% Then, we show that the computation of polynomial invariants of a loop up can be reduced to the one of some invariant set. 
to %efficiently 
decide if a given polynomial is invariant (Proposition~\ref{prop3.4}).

\item We design two algorithms for computing %the set of 
polynomial invariants  of a loop up to a fixed degree. The first one (Theorem~\ref{algogeneral}), when the initial value is not fixed, outputs a linear parametrization which depends polynomially on this value. The second (Algorithm~\ref{alg2}), when the initial value is fixed, is much more efficient and computes a basis of this set, seen as a vector space. Experiments with our prototype implementation demonstrate the practical efficiency of our algorithms, solving problems beyond the current state of the art. 
%While we have focused on loops whose guards are all polynomial equations, 
Note that the algorithms can be adjusted to include disequalities in the guard;
see Remarks\,\ref{remark_disequality}\,and\,\ref{remark_disequality_fixed}.
% , by means of a linear parametrization, 
% whose coefficients depend polynomially on the initial values. 

% Furthermore, we establish a criterion for determining whether a given polynomial is invariant for a given loop, facilitating the algorithmic computation of the invariant ideals, in Theorem~\ref{algogeneral}.

\item %In Section~\ref{sec:lift}, 
Finally, we apply these algorithms to other problems: we show how to lift some polynomial invariants for the non-fixed initial value case from the fixed one (Proposition~\ref{general}); we consider the case with inequalities in the loop (Proposition~\ref{prop:saloop}).
%\item We present a procedure for computing polynomial invariants up to a specified degree using invariant sets (Algorithm~\ref{alg2}).

%\item \textcolor{red}{add citation to the new theorem}\textcolor{blue}{Add that the guard polynomials $g_i$'s and the initial values are not used in the previous studies.}

%\item \todoremigh{remove?}{Our approach identifies a minimal generating set for the invariant ideal of a specified degree. Combinations of these polynomials serve as polynomial invariants as well. In contrast, previous works, such as in \cite{Unsolvableloops}, randomly generated polynomials within this ideal, without guaranteeing minimal degree.}

%\item Note that, while it is straightforward to find \todoremigh{what is that?}{closed formulas} given polynomial invariants, the reverse process is only feasible under very strong assumptions, as detailed in \cite{Unsolvableloops}. Specifically, one needs to identify linear relations among certain program variables, which is a challenging task in itself. \textcolor{blue}{I merged this with the paragraph below. Please double check if it makes sense.}
\end{enumerate}
}

%%\cite{manna2012temporal}The property 2 of the following paper is about linearizable property of solvable mappings: \cite{de2016polynomial}
\noindent\textbf{Related works.}
A common approach for generating polynomial invariants entails creating a system of recurrence relations from a loop, acquiring a closed formula for this recurrence relation, and then computing polynomial invariants by removing the loop counter from the obtained closed formula (as in \cite{rodriguez2007automatic}). Note that it is straightforward to find such recursion formulas from a polynomial invariant. However, the reverse process is only feasible under very strong assumptions, as detailed in \cite{Unsolvableloops}. Specifically, one needs to identify polynomial relations among  program variables, which is a challenging task in itself. 

In \cite{hrushovski2018polynomial}, an algorithm is designed to compute the Zariski closure of points generated by affine maps. Another perspective, detailed in \cite{Unsolvableloops}, categorizes variables into effective and defective sets, where closed formulas can be computed for effective variables but not for defective variables. %The approach of \cite{Unsolvableloops} strives to identify polynomials in defective variables that allow closed formulas. With closed formulas accessible for both effective and defective variables, the elimination of the loop counter becomes adequate for generating polynomial invariants.
Similarly, the methodology proposed in \cite{kovacs2008reasoning} is specifically tailored for P-solvable loops.
In \cite{cyphert2023solvable}, the method of approximating a general program by a solvable program is discussed. %Such an approximation makes 
This approach is incomplete, however, a monotonicity property is proven ensuring that such an approximation can be improved as much as required. %\textcolor{red}{In \cite{cyphert2023solvable}, the method of converting loops into solvable loops is discussed, though such conversion may not always be feasible for non-zero solvable loops.} %As mentioned earlier, the successful generation of all polynomial invariants for solvable loops has been demonstrated.
 %This approach entails two key steps: 1.) establishing closed formulas that encompass the initial values of program variables, the loop counter, and new variables, taking into account algebraic relations among the new variables, and 2.) eliminating the new variables and the loop counter to derive polynomial invariants among program variables.
In \cite{muller2004computing}, Müller-Olm and Seidl employ ideas similar to ours in Algorithm 1, with the notable difference that, through our geometric approach for computing the invariant set, we have established a better stopping criterion by comparing the equality of radical ideals rather than the ideals themselves. They also impose algebraic conditions on the initial values and subsequently compute polynomial invariants that need to apply to all initial values satisfying these constraints. Consequently, polynomial invariants that are applicable to all but e.g. finitely many such initial values might be overlooked.
In contrast, our algorithm %\textsf{TruncatedInvariant},
outlined in Theorem 3.5, yields polynomial invariants that depend on the initial values, addressing a much broader problem, which, to our knowledge, has not been previously tackled.
Moreover, Algorithm 2, which tackles cases with fixed initial values, is significantly faster than Algorithm 1 and its counterpart in \cite{muller2004computing}. Additionally, in Proposition 4.1, we outline a comprehensive procedure that relies on Algorithm 2 to identify general invariants of specific form %the form $f(x)-f(a)$ 
whenever they exist, enabling us to generate invariants produced in \cite{Unsolvableloops,amrollahi2023solvable}.

Finally, the case of polynomial invariants represented as inequalities has been considered in \cite{chatterjee2020polynomial}, using tools such as Putinar's Positivstellens\"atz. However, this presents a different problem, %involving much stronger polynomial invariants because geometrically,
involving semi-algebraic sets $X$ whose image $F(X)$ is a subset of $X$. However, polynomial invariants do not necessarily satisfy this property. %For example, we computed a basis of the first truncated invariant ideal of yagzhev9, and each basis vector does not possess this property.

\begin{comment}
    
\smallskip\noindent\textbf{Structure of the paper.}
% In Section~\ref{sectionmath}, we briefly recall the basic algebraic notions use throughout the paper. 
In Section~\ref{sectioninvariantset}, we introduce our approach to computing invariant sets using algebraic tools, Algorithm~\ref{algo1}. The process of generating polynomial invariants using these invariant sets is detailed in Section~\ref{sectiongenerating}, accompanied by illustrative examples. In Section~\ref{sec:application}, we demonstrate some applications of Algorithm~\ref{alg2} in various examples from the literature. In Section~\ref{sec:implementation}, we present tables summarizing our experiments and implementation details.
 
\end{comment}

% Throughout the paper, well-known results and standard mathematical notions are used; hence, their proofs are omitted.

% \section{Mathematical background}\label{sectionmath}
% In this section, we establish our terminology and review key concepts from algebraic geometry utilized throughout the paper. 
%\begin{proposition}[\textcolor{blue}{~\cite[Page 232]{cox2013ideals}}]\label{prop:primaryideal}\todoremi{not explicitly used anywhere}
%Let $X\subset \mathbb{C}^n$ be an algebraic variety. Then, $I(X)$ has a unique decomposition  into primary ideals $I(X) = Q_1 \cap \cdots \cap Q_n$. Moreover, $X = \V(Q_1) \cup \cdots \cup \V(Q_n)$ is an irreducible decomposition of $X$.
%\end{proposition}

%%%%%%%%%%%%%%%%%%%%%%%%%%%%%%%%%%%%%%%%%%%%%%%%%%%%%%%%%%%%%%%%%%%%%%%%
\section{Computing invariant sets}\label{sectioninvariantset}
%%%%%%%%%%%%%%%%%%%%%%%%%%%%%%%%%%%%%%%%%%%%%%%%%%%%%%%%%%%%%%%%%%%%%%%%
We will first establish an effective description of the invariant set associated with a given algebraic variety and a polynomial map. Subsequently, we will derive an algorithm based on this description to compute such a set. 
We begin with a technical lemma to express the preimage of an algebraic variety under a polynomial map.
%In this paragraph, we prove an effective description of the invariant set associated to a given algebraic variety and a polynomial map. Then, we derive an algorithm based on this description to compute such set. We start with a technical lemma providing an effective description of the preimage of an algebraic variety with respect to a polynomial map.
\begin{lemma}\label{lem:equtionspreimage}
Given a polynomial map $F: \mathbb{C}^n \longrightarrow \mathbb{C}^m$ and  an algebraic variety $X \subset \mathbb{C}^m$, the preimage $F^{-1}(X)$ is also an algebraic variety. Moreover, if $X=\V(g_1, \ldots, g_k)$ and $F = (f_{1}, \ldots, f_{m})$, where $g_1,\ldots,g_k\in \mathbb{C}[y_1, \ldots, y_m]$ and $f_1, \ldots, f_m \in \mathbb{C}[x_1, \ldots, x_n]$, then 
\[
F^{-1}(X) = \V(g_{1}(f_1, \ldots , f_m), \ldots, g_{k}(f_1, \ldots, f_m)) \subset \mathbb{C}^n. 
\]
\end{lemma}
\vspace{-2mm}
\begin{proof}
%Let $F = (f_{1}, \ldots, f_{m})$ and $X = \V(g_1, \ldots, g_k)$ as in the statement, 
By definition, we have that:
  \begin{align*}
    F^{-1}(X) %&= %\{ x \in \C^n \mid(f_{1}(x), \ldots, f_{m}(x)) \in \V(g_1,\ldots,g_k) \}\\
 &= \{ x \in \C^n \mid \forall\, 1 \leq i \leq k,\, 
 g_i(f_1(x),\ldots,f_m(x)) = 0\}.
  \end{align*}
 Let $h_i =  g_i(f_1,\ldots,f_m)\in\C[x_1,\ldots,x_n]$, for all $1 \leq i \leq k$, then
 $
    F^{-1}(X) = \V(h_1,\ldots,h_k),
 $ 
 and $F^{-1}(X)$ is an algebraic variety.
\end{proof}

Before proving the main result, we need the following lemma.
\begin{lemma}\label{lem:invset}
Let $F : \Cn \longrightarrow \Cn$ be a polynomial map and $X\subseteq\Cn$ an algebraic variety. Then, $F(S_{(F,X)})$ is a subset of $S_{(F,X)}$.
\end{lemma}
\begin{proof}
 Let $x \in S_{(F,X)}$. By the definition of the invariant set, $F^{(m)}(x) \subseteq X$ for every $m$. Thus, $F^{(m)}(F(x)) \subseteq X$, implying that $F(x) \in S_{(F,X)}$ for every $x \in S_{(F,X)}$. Hence, $F(S_{(F,X)}) \subseteq S_{(F,X)}$.
 %Let $x$ be an element in $S_{(F,X)}$. By the definition of the invariant set, $F^{(m)}(x)$ is contained in $X$ for every integer $m\in \N$. Therefore, $F^{(m)}(F(x))$ is contained in $X$, % for every integer $m\in \N$ which proves that $F(x)$ is contained in $S_{(F,X)}$ for every $x \in S_{(F,X)}$. Thus, $F(S_{(F,X)})$ is a subset of $S_{(F,X)}$.
\end{proof}

\vspace{-2mm} We now give an effective method to compute invariant sets, by means of a stopping criterion for the intersection of the iterated preimages. In the following, %for any $m\geq 0$, and any $Y \subset \C^n$, 
$(F^{(m)})^{-1}$ will be denoted by $F^{(-m)}$.
\begin{proposition}\label{prop:stabilization}
Let $F : \Cn \longrightarrow \Cn$ be a polynomial map and $X\subseteq\Cn$ an algebraic variety. We define $X_m=\bigcap\limits_{i=0}^{m}F^{-i}(X)$ for all $m \in \N$. Then, the following statements are true:
\begin{itemize}
        \item[$(a)$] $X_{m+1} \subseteq X_{m}$ for all $m$.
        \item[$(b)$] There exists %a natural number 
        $N\in\mathbb{N}$ such that $X_{N} = X_{m}$ for all $m \geq N$.
        \item[$(c)$]  If $X_{N}=X_{N+1}$ for some $N$, then $X_{N}=X_{m}$ for all $m\geq N$.
        \item[$(d)$]  The invariant set $S_{(F,X)}$ is equal to $X_{N}$.
    \end{itemize}
\end{proposition}
\begin{proof}
 $(a)$ The following is straightforward from the definition: 
 $$X_{m+1}=X_m\cap F^{-(m+1)}(X)\subseteq X_m.$$
 
        \smallskip
        
        $(b)$ From $(a)$, we have the following descending chain 
        \[ 
            X_{0} \supseteq X_{1} \supseteq X_{2} \supseteq \cdots\supseteq X_{m} \supseteq X_{m+1} \supseteq \cdots,
        \]
       which are algebraic varieties by Lemma~\ref{lem:equtionspreimage}. 
       Thus, we have: % the ascending chain of ideals: 
       \[
        I(X_{0})\subseteq I(X_{1}) \subseteq I(X_{2})\subseteq \cdots \subseteq I(X_{m})\subseteq I(X_{m+1})\subseteq \cdots.
        \]
    Since $\C[x_{1},x_{2},\ldots,x_{n}]$ is a Noetherian ring, there exists a natural number $N$ such that $I(X_{N})=I(X_{m})$ for all $m\geq N$. Therefore, 
    $$X_N = \V(I(X_N)=\V(I(X_m))=X_m \text{ for all } m\geq N.$$

\smallskip 
$(c)$ For such an $N$, we have that 
\[
    X_{N+2} =X\cap F^{-1}(X_{N+1})=X\cap F^{-1}(X_N)=X_{N+1}.
 \] 
Thus,
$
X_{m}=X_{m+1}$ for all $m \geq N$,
and so %which implies 
$X_N = X_m$ for all $m\geq N$.

\smallskip
$(d)$ We will first prove that $S_{(F,X)}\subseteq X_{m}$ for every $m$, by induction on $m$. 
By the invariant set's definition, $S_{(F,X)}$ is a subset of $X=X_{0}$ which proves the base case $m=0$.
Now let $m>0$ and assume $S_{(F,X)}\subseteq X_{m-1}$.
%\textbf{Induction step.} 
By Lemma~\ref{lem:invset} and the induction hypothesis,
$$F(S_{(F,X)})\subset S_{(F,X)}\subset X_{m-1}.$$
Therefore, $S_{(F,X)}$ is a subset of $F^{-1}(X_{m-1})$. Note that $S_{(F,X)}$ is a subset of $X$ by the definition. Thus, 
$S_{(F,X)}\subset F^{-1}(X_{m-1})\cap X = X_{m}.$
In particular, when $m=N$, we have that $S_{(F,X)}\subseteq X_{N}$.

\smallskip
To prove the other inclusion, for every $x\in X_{m}$, note that $F^{m}(x)$ is contained in $X$ since $x\in X_m\subseteq F^{-m}(X)$. By $(a)$ and $(b)$, $X_{N}$ is contained in $X_{m}$ for every $m\in N$. Thus, $F^{m}(x)$ is contained in $X$ for every $x\in X_{N}$ and every $m\in \N$. Hence, $X_{N}\subseteq S_{(F,X)}$.
\end{proof}

\begin{samepage}
\begin{remark}\label{remark1}
%\begin{itemize}
  %  \item[(a)] 
  By Theorem~\ref{prop:stabilization}(d), the invariant set $S_{(F,X)}$ is an algebraic variety, since by construction each $X_{i}$ is an algebraic variety.
  %  \item [(b)] 
  By Theorem~\ref{prop:stabilization}(a), the ideal of $X_{j}$ is a subset of the ideal of $X_{i}$ for every $i \geq j$. Hence, although computing the ideal of $X_{N}$ where $X_{N}=X_{N+1}$ may be infeasible, leveraging computable $X_{i}$'s for $i<N$ provides partial information.
%\end{itemize}
\end{remark}
\end{samepage}
\vspace{-1mm}
We now present an algorithm for computing the invariant set associated to an algebraic variety and a polynomial map, described by sequences of multivariate polynomials. We restrict here to \emph{rational} coefficients as this covers the target applications, and we need to work in a computable field for the sake of the effectiveness.
%\todoremi[inline]{Algorithm to be written in a computer science language and precise input/output}

\begin{algorithm}
\caption{\InvariantSet}\label{algo1}
\begin{algorithmic}[1]\setstretch{1.1}
\Require Two sequences $\gb$ and $F = (f_1,\ldots, f_n)$ in $\mathbb{Q}[x_1,\ldots,x_n]$.
\Ensure Polynomials whose common zero-set is $S_{(F,{\V(\gb)})}$.
\State $S \gets \{\gb\};$
% \State $S_{GB} \gets \GB(\gb);$
\State $\gbt \gets \compose(\gb,\,F);$
% \State $S' \gets S \cup \compose(G,F)$
% \State $i\gets 1;$
% \For{$i\leq k$}
% \State$g_i\gets g_i(f_{1},\ldots, f_{n});$
% \State$i\gets i+1;$
% \EndFor
% \State$G_{1}\gets G_{1}.\text{extend}((g_{1},\ldots, g_{k}));$
% \State$N\gets 0$;
\While{ $\InRadical(\gbt ,\,S)==\texttt{False}$}
\State $S \gets S \cup \{\gbt\};$
% \State $S_{GB} \gets \GB(S_{GB} \cup \{\gbt\});$
\State $\gbt \gets \compose(\gbt,\,F);$
% \State$N\gets N+1;$
% \State$G_{N+1}\gets G_{N};$
% \State $i\gets 1;$
% \For{$i\leq k$}
% \State$g_i\gets g_i(f_{1},\ldots, f_{n});$
% \State$i\gets i+1;$
% \EndFor
% \State$G_{N+1}\gets G_{N+1}.\text{extend}((g_{1},\ldots, g_{k}));$
\EndWhile
\State \Return $S$;
% \State \Return N
\end{algorithmic}
\end{algorithm}
In Algorithm~\ref{algo1}, the procedure ``\compose'' takes as input two sequences of polynomials $\gb=(g_1,\dotsc,g_k)$ and $F=(f_1,\dotsc,f_n)$ in $\Q[x_1,\dotsc,x_n]$ and outputs a sequence of polynomials $(h_1,\dotsc,h_k)$ in $\Q[x_1,\dotsc,x_n]$, such that $h_i = g_i(f_1,\dotsc,f_n)$ for all $1\leq i \leq k$.

% The procedure \GB takes as input a finite set of polynomials $\bg$ in $\Q[x_1,\ldots,x_n]$ and computes a Groebner basis (w.r.t. some monomial order) of the ideal of $\Q[x_1,\ldots,x_n]$ generated by $\bg$.

The procedure ``\InRadical'' takes as input a sequence $\gbt$ and a set $S$ both in $\Q[x_1,\dotsc,x_n]$ and decides if all the polynomials in $\gbt$ belong to the \emph{radical} of the ideal generated by $S$.
By \cite[Chap 4, \S 2, Proposition 8]{cox2013ideals}, the latter procedure can be performed by computing a Gr\"obner basis for the ideal of $\C[x_1,\dotsc,x_n,t]$, generated by $1-t\cdot \gbt$ and $S$, where $t$ is a new variable.

\smallskip
We now prove the termination and correctness of Algorithm~\ref{algo1}.

\begin{theorem}\label{thm:corralgo1}
    On input two sequences $\gb=(g_1,\ldots, g_k)$ and $F=(f_1,\ldots,f_n)$ of polynomials in $\Q[x_1,\ldots,x_n]$, Algorithm~\ref{algo1} terminates and outputs a sequence of polynomials whose vanishing set is the invariant set $S_{(F, \V(g_1,\ldots, g_k))}$. %\textcolor{red}{What is $J$?}
\end{theorem}
\begin{proof}
Consider the algebraic variety $V=\V(\gb)$ and the polynomial map $F=(f_1,\ldots,f_n):\Cn\longrightarrow \Cn$. Let $S_0 = \gb$, and let $S_m$ denote the set $S$ after completing $m\geq 1$ iterations of the \textbf{while} loop in Algorithm~\ref{algo1}. Similarly, let $\gbt_0=\gb$, $\gbt_1=\gb(F)$, and $\gbt_{m+1}$ be the sequence $\gbt$ after $m$ iterations.
%Let $I_{m}$ be the ideal generated by $S_m$. Since $S_{m}$ is contained in $S_{m+1}$, the ideal $I_{m}$ is contained in $I_{m+1}$.  Therefore, we have the following ascending chain of ideals:
%\[ 
%             \sqrt{I_{0}} \subseteq  \sqrt{I_{1}} \subseteq  \sqrt{I_{2}} \subseteq \cdots\subseteq  \sqrt{I_{m}} \subseteq \sqrt{I_{m+1}} \subseteq \cdots.
%\]
%Since $\C[x_1,\ldots, x_n]$ is a Noetherian ring, the above chain becomes stationary which verifies that Algorithm~\ref{algo1} terminates. 
Let $m\geq0$, and as in Proposition~\ref{prop:stabilization}, let $X_m=\bigcap\limits_{i=1}^{m}F^{-i}(X)$.
%We first prove that $X_{m}=\V(S_{m})$ for every $m\geq 0$.  
By construction,
$
S_{m} \,=\, \left\{\gbt_0,\dotsc, \gbt_{m}\right\} $, that is $S_m\,=\,\left\{\gb, \gb(F),\ldots, \gb(F^{m})\right\}
$, and 
so by Lemma~\ref{lem:equtionspreimage},
\begin{equation*}
  % \begin{split}
    X_{m} \,=\, \bigcap\limits_{i=0}^{m}F^{-i}(\V(\gb)) \,=\, \bigcap\limits_{i=0}^{m}\V(\gb(F^{i}))
    \,=\, \V(S_{m}).
   % \end{split}
\end{equation*} 
%which implies that $X_m=\V(I_m)$ for any $m\geq 0$. 
By Proposition~\ref{prop:stabilization}.(b), there exists $N\in\N$ such that $X_{N}=X_{N+1}$, that is $\V(S_N)=\V(S_{N+1})$. %by above. 
This means that the polynomial 
$\gbt_{N+1} = \gb(F^{N+1})$ vanishes on $\V(S_N)$, or equivalently by the Hilbert's Nullstellensatz \cite[Chap 4, \S 1, Theorem 2]{cox2013ideals}, that $\gbt_{N+1}$ belongs to $\sqrt{I(S_N)}$.

Hence, Algorithm~\ref{algo1} terminates after $N$ iterations of the \textbf{while} loop and outputs $S_N$. In particular, by Proposition~\ref{prop:stabilization}.(d), 
$S_{(F,X)} = X_N = \V(S_N)$, which proves the correctness
%This proves the correctness 
of Algorithm~\ref{algo1}.
\end{proof}

\begin{remark}The complexity analysis of Algorithm~\ref{algo1} is not detailed in this paper, as the worst-case complexity bounds given by the literature are very pessimistic. The first main issue concerns the number of loop iterations performed by Algorithm~\ref{algo1}, which can exhibit a growth behavior similar to Ackermann’s function \cite{MorenoSocias_1992,pastuszak2020ascending}. Furthermore, the radical membership test after each iteration involves Gröbner bases computations, which can have a complexity doubly exponential in the number of variables for some tailored examples \cite{mayr1982complexity}. In practice, as discussed in, for example, \cite[\S 21.7]{von2013modern}, these algorithms show reasonable costs and benefit from active research \cite{eder2017survey} and efficient implementations \cite{berthomieu2021msolve}.

Despite all of this, the experimental section shows that this algorithm can be applied in practice to loops found in the literature. Moreover, future work includes accounting for the specific structure of loops %systems 
to enhance both practical and theoretical efficiency.
\end{remark}
\vspace{-3.5mm}
%%%%%%%%%%%%%%%%%%%%%%%%%%%%%%%%%%%%%%%%%%%%%%%%%%%%%%%%%%%%%%%%%%%%%%%%
\section{Generating polynomial loop invariants}\label{sectiongenerating}
%%%%%%%%%%%%%%%%%%%%%%%%%%%%%%%%%%%%%%%%%%%%%%%%%%%%%%%%%%%%%%%%%%%%%%%%

We first fix our notation throughout this section. In the polynomial ring $\C[x_1, \dots , x_n]$, we fix the notation ${\bf x}^\alpha$ with $\alpha=(a_1,\ldots,a_n)\in \mathbb Z_{\ge 0}^n$ denoting the monomial $x_1^{a_1}\dots x_n^{a_n}$. 
Throughout when we write $f= b_{1}{\bf x}^{\alpha_{1}}+\cdots+b_{m}{\bf x}^{\alpha_{m}}$, we refer to the expression of $f$ in the basis of monomials, where $f$ consists of exactly $m$ terms (or monomials) $b_i{\bf x}^{\alpha_{i}}$ with coefficients $b_i\in \mathbb{C}$.
In our polynomial expression, we always order the monomials such that for $i < j$:
\[
\deg({\bf x}^{\alpha_i}) < \deg({\bf x}^{\alpha_j}) \text{ or } \left(\deg({\bf x}^{\alpha_i}) = \deg({\bf x}^{\alpha_j}) \text{ and } {\bf x}^{\alpha_i} >_{\text{lex}} {\bf x}^{\alpha_j}\right)
\]
where $\deg({\bf x}^{\alpha_i})$ represents the degree of the monomial ${\bf x}^{\alpha_i}$, and the lexicographic order is with respect to the order of the variables $x_1 > x_2 > \cdots > x_n$. We also denote $|\alpha_i|$ for the size of the vector $\alpha_i=(\alpha_{i,1},\ldots,\alpha_{i,n})$ which is $\alpha_{i,1}+\cdots+\alpha_{i,n}$. 
 
%We will denote by ${\rm span}(v_1 ,\ldots, v_n)$ the vector space generated by $v_{1},\ldots, v_{n}$. 
%We use bold letters to denote vectors (for instance $\bf x_0$, $\bf a_{i_1,\ldots i_n}$).

\subsection{The general case}
\begin{definition}\label{def:loop}
Let $\ab\in\mathbb{C}^n$, $\gb=(g_1,\ldots, g_k)$ and $F=(f_1,\ldots,f_n)$ be two sequences of polynomials in $\C[x_1,\ldots,x_n]$. Consider the algebraic variety $X=\V(\gb)$ and the polynomial map $F= (f_1,\ldots,f_n)$.
Then $\mathcal{L}(\mathbf{a},\gb,F)$ (or $\mathcal{L}(\ab,X,F)$) denotes the polynomial loop on Page~\ref{page:alg}.
%\programbox[0.55\linewidth]{
%\State$(x_1, x_2,\ldots, x_n)=(a_1,a_2,\ldots,a_n)$
%\While{$g_1 = g_2 =\cdots = g_k =0$}
%\State $\begin{pmatrix}
%x_1 \\
%x_2 \\
%\vdots \\
%x_n
%\end{pmatrix}
%\xleftarrow{\textbf{F}}
%\begin{pmatrix}
%f_1\\
%f_2\\
%\vdots\\
%f_n
%\end{pmatrix}
%$
%\EndWhile
%}

% \end{samepage}
 \noindent  
% which is called the polynomial loop with initial value $\ab$, on $V$, and with respect to $F$. 
When no $\bf g$ is identified, we will write $\mathcal{L}(\mathbf{a},0,F)$.
Finally, we will simply write $\mathcal{L}$ when it is clear from the context.
\end{definition}
%\todo{I think this is the notation you want for Prop 3.4?}

\begin{proposition}\label{prop:terminateinvariant}
    Let $\ab\in\mathbb{C}^n$, $X$ be an algebraic variety and $F:\C^n\to \C^n$ a polynomial map. 
    Then, the polynomial loop $\mathcal{L}(\ab,X,F)$ never terminates if and only if $\ab \in S_{(F, X)}$.
\end{proposition}
\begin{proof} 
% Let $\ab= (a_1, a_2,\ldots, a_n)$ be the initial value of $\mathcal{L}$. 
The statement directly follows from the definition, as $\mathcal{L}(\ab,X,F)$ never terminates if, and only if, $F^{(m)}(\ab)\in X$ for all $m\geq 0$, %\tag*{\qedhere}
%\end{center}
that is, if and only if $\ab\in S_{(F,X)}$.
\end{proof}

\vspace{-3mm}\begin{example} %\textcolor{red}{Let's} $\rightarrow$ \textcolor{blue}{correct: Let us}
Let us compute the termination condition for the following loop $\mathcal{L}$ where $F=(f_1,f_2)=(10x_1-8x_2, 6x_1-4x_2)$, $g=x_1^2-x_1x_2+9x_1^3-24x_1^2x_2+16x_1x_2^2$, and $X=\V(g)$.

\programbox[0.53\linewidth]{
\State$(x_1, x_2)=(a_1,a_2)$
\While{$g=0$}
\State $\begin{pmatrix}
x_1 \\
x_2
\end{pmatrix}
\xleftarrow{\textbf{F}}
\begin{pmatrix}
10x_1-8x_2\\
6x_1-4x_2
\end{pmatrix}
$
\EndWhile
}
\\
Algorithm~\ref{algo1} computes the invariant set %$S_{(F,X)}$ 
through the following steps:
\begin{itemize}
\item Initially, $S$ is set to ${g}$, and $\tilde{g}=\compose(g,F)=360x_1^3-1248x_1^2x_2+40x_1^2+1408x_1x_2^2-72x_1x_2-512x_2^3+32x_2^2$.
\item By computing a Gr\"obner basis for the ideal generated by $g$ and $1-t\tilde{g}$, it is determined that $\InRadical(\tilde{g}, S)=\text{False}$.
\item The set $S$ is then updated to include $\tilde{g}$, resulting in $S=S\cup\{\tilde{g}\}=\{x_1^2-x_1x_2+9x_1^3-24x_1^2x_2+16x_1x_2^2,\, 360x_1^3-1248x_1^2x_2+40x_1^2+1408x_1x_2^2-72x_1x_2-512x_2^3+32x_2^2\}$, and $\tilde{g}$ is recomputed as $\compose(\tilde{g},F)=7488x_1^3-26880x_1^2x_2+832x_1^2+31744x_1x_2^2-1600x_1x_2-12288x_2^3+768x_2^2$.
\item This time, the computation yields $\InRadical(\tilde{g}, S)=\text{True}$.
\end{itemize}
Thus, the output of Algorithm~\ref{algo1} is given by
%\begin{flushleft}
$\{x_1^2-x_1x_2+9x_1^3-24x_1^2x_2+16x_1x_2^2\textbf{,}60x_1^3-1248x_1^2x_2+40x_1^2+1408x_1x_2^2-72x_1x_2-512x_2^3+32x_2^2\}.$
%\end{flushleft}
%This set represents the polynomials iteratively generated by the algorithm. 
The radical of the ideal generated by this output is generated by %the polynomial 
$h:=x_1-x_2-9x_1^2+24x_1x_2-x_2^2$. Consequently, by Proposition~\ref{prop:terminateinvariant}, $\mathcal{L}$ never terminates if and only if $(a_1, a_2)\in \V(h)$. %Thus, the conditions for the termination of $\mathcal{L}$ are precisely defined by the solutions to the polynomial $h$. %equation $x_1-x_2-9x_1^2+24x_1x_2-x_2^2$.
%Therefore, The output of Algorithm~\ref{algo1} is 
%\begin{flushleft}
  %  $\{x_1^2-x_1x_2+9x_1^3-24x_1^2x_2+16x_1x_2^2\textbf{,}60x_1^3-1248x_1^2x_2+40x_1^2+1408x_1x_2^2-72x_1x_2-512x_2^3+32x_2^2\}.$
%\end{flushleft}
%The radical of the ideal generated by the output is the ideal generated by $x_1-x_2-9x_1^2+24x_1x_2-x_2^2$. Thus, by Proposition~\ref{prop:terminateinvariant} we have that $\mathcal{L}$ never terminates if and only if $(a_1, a_2)\in \V(x_1-x_2-9x_1^2+24x_1x_2-x_2^2)$.
\end{example}

%\vspace{-3mm}

\vspace{-2mm}
\begin{definition}
Polynomial invariants of a loop $\mathcal{L}$ are polynomials that vanish before and after every iteration of $\mathcal{L}$. 
The set $I_{\mathcal{L}}$ of all polynomial invariants for $\mathcal{L}$ is an ideal, called the \emph{invariant~ideal~of~$\mathcal{L}$}.
Let $d \geq 1$, the subset $I_{d,\mathcal{L}}$ of all polynomial invariant for $\mathcal{L}$, of total degree $\leq d$, is called the dth \emph{truncated invariant ideal of $\mathcal{L}$}.
    % is the dth truncated ideal of   $I_{\mathcal{L}}$.
    % For an ideal $I\subset\C[x_1,\ldots,x_n]$, the \emph{dth truncated ideal} $I_{d}$ is the set $\{f\in I\mid \deg(f)\leq d\}$. 
\end{definition}

% \begin{remark}\label{remark2}
% The invariant polynomial ideal $I_{\mathcal{L}}$ is indeed an ideal. As mentioned earlier, the truncated ideals can be equipped with a vector space structure. Thus, to obtain all polynomial invariants of degree $d$, it 
Though a truncated invariant ideal is not an ideal, it has the structure of a finite~dimensional~vector~space. Hence, it can be uniquely parametrized by a system of linear equations, whose coefficients depend on the initial values. 
% The solution set of this linear system corresponds to the dth truncated invariant ideal $I_{d,\mathcal{L}}$ of a loop $\mathcal{L}$.
In the following, we demonstrate how to reduce the computation of such a parametrization for a given loop to computing an invariant set of an extended polynomial map. % defined from this loop.

% \end{remark}
\smallskip
We start by a criterion to determine whether a given polynomial is invariant with respect to a given loop or not.

%The invariant polynomial ideal $I$ is actually an ideal. As we mentioned before, the truncated ideals are vector spaces. Therefore, it is sufficient to compute a basis for the $d^{th}$ truncated invariant ideal $I_{d,\mathcal{L}}$ of a loop $\mathcal{L}$.

% \begin{example} 
%\begin{samepage}
%For the following loop, we have a polynomial invariant of degree 2 
%\[ 3x_1-4x_2+3x_1x_2+9x_1^2-20x_2^2=0.\]
%\begin{algorithmic}
%\State$(x_1, x_2)=(4,3)$
%\While{true}
%\State $\begin{pmatrix}
%x_1 \\
%x_2 
%\end{pmatrix}
%\longleftarrow 
%\begin{pmatrix}
%9x_1-4x_2\\
%3x_1+2x_2
%\end{pmatrix}
%$

%\EndWhile
%\end{algorithmic}
%\end{samepage}
% \end{example}

% \begin{definition}
%     Let $\mathcal{L}$ be a polynomial loop with a polynomial map $F=(f_1,\ldots,f_n)$. The $m^{\text{th}}$ extended loop $\mathcal{L}_{m}$ of $\mathcal{L}$ is a polynomial loop with $F_{m}(x_1,\ldots, x_n, y_1,\ldots, y_m)=(f_1(x_1,\ldots, x_n),\ldots,f_n(x_1,\ldots, x_n), y_1,\ldots,y_m)$.
% \end{definition}

\begin{proposition}\label{prop3.4}
Let ${\bf a} \in \C^n$ and $F=(f_1,\ldots,f_n) \subset \C[x_1,\ldots,x_n]$.
For $m = \binom{n+d}{d}$, let ${\bf b} \in \C^m$ and $g({\bf x, y}):=\sum_{ |\alpha_i|\leq d}y_{i}{\bf x}^{\alpha_i}$ be a degree $d$ polynomial in $\C[x_1,\ldots,x_n, y_1,\ldots,y_m]$. 
Then $g(\xb,{\bf b})$ is a polynomial invariant for $\mathcal{L}({\bf a},0,F)$
if, and only if,  $({\bf a, b}) \in S_{(F_{m},X)}$
where 
 $F_m = (f_1,\ldots,f_n,y_1,\ldots,y_m)$ and $X = \V(g) \subset \C^{n+m}$.
\end{proposition}
\begin{proof}
Consider the following ``mth extended'' loop $\mathcal{L}((\ab,\bb),g,F_{m})$:

    \programbox[0.8\linewidth]{
\State$(\xb, \yb)=(\ab, \bb)$
\While{$g(\xb,\yb)=0$}
\State $\begin{pmatrix}
x_1 \\
\vdots \\
x_n\\
y_1\\
\vdots\\
y_m
\end{pmatrix}
\longleftarrow 
\begin{pmatrix}
f_1(x_1,\ldots,x_n)\\
\vdots\\
f_n(x_1,\ldots,x_n)\\
y_1\\
\vdots\\
y_m
\end{pmatrix}
$
\EndWhile
}

Let $\ab^0=\ab$ and for $k\geq 1$ let $\ab^k = F(\ab^{k-1})$. Then, $(\ab^k)_{k\in \N}$ are the successive values of $\xb$ in $\mathcal{L}(\ab,0,F)$.
Let $\bb \in C^m$ and assume that $g(\xb,\bb)$ is a polynomial invariant for $\mathcal{L}(\ab,0,F)$. Let $k\geq0$, then after the $k^{\text{th}}$ iteration of the extended loop $\mathcal{L}((\ab,\bb),g,F_{m})$, the value of $\xb$ is $\ab^k$ and the value of $\yb$ is still $\bb$. Since, by assumption $g(\ab^k,\bb)=0$, this loop does not stop after the $k^{\text{th}}$ iteration and, by induction never terminates. The converse is immediate.

Therefore, $g(\xb,\bb)$ is a polynomial invariant for $\mathcal{L}$ if and only if the extended loop $\mathcal{L}_{m}$ never terminates, which is equivalent, by Proposition~\ref{prop:terminateinvariant}, to $(\ab,\bb)\in S_{(F_{m},\V(g))}$.
\end{proof}

\noindent The following main result follows from the above criterion.
\begin{theorem}\label{algogeneral}
 Let $F=(f_1,\ldots,f_n)$ be a sequences of polynomials in $\Q[x_1,\ldots,x_n]$ and let $d \geq 1$ and $m = \binom{n+d}{d}$.  Then, there exists an algorithm \textsf{TruncatedInvariant} which, on input $(F,d)$ computes a polynomial matrix $A$, with $m$ columns, and coefficients in  {$\Q[x_1,\ldots,x_n]$}, such that %the following holds.
the dth truncated invariant ideal of $\mathcal{L}(\ab,0,F)$~for any $\ab \in \Q^n$ is: % equal to:
 %Then for any $\ab \in \Q^n$, the dth-truncated invariant ideal of $\mathcal{L}(\ab,0,F)$ is defined by
 \[
  I_{d,\mathcal{L}} = 
  \left\{ \sum_{ |\alpha_i|\leq d}b_{i}{\bf x}^{\alpha_i} \mid (b_1,\ldots,b_m) \in \ker\,A(\ab) \right\}
 \]
  where~$\ker\,A(\ab)$~is~right-kernel of $A$, whose entries~are~evaluated~at~$\ab$.
\end{theorem}
\begin{proof}
    Let $y_1,\ldots,y_m$ be new indeterminates, and define $g$, $F_m$ and $X$ as in Proposition~\ref{prop3.4}. Then, by Proposition~\ref{thm:corralgo1}, on input $(g,F_m)$, Algorithm~\ref{algo1} computes polynomials $h_1,\dotsc,h_N$ in $\Q[x_1,\ldots,\\ x_n,y_1,\dotsc,y_n]$, whose common vanishing set is $S_{(F_m,X)}$.
    Moreover, by construction of Algorithm~\ref{algo1} and definition of $F_m$, we have:
    \[
        h_j = g\circ F_m^j(x_1,\ldots,x_n,y_1,\dotsc,y_m) = g\Big(F^j(x_1,\ldots,x_n),y_1,\dotsc,y_m\Big)
    \] 
     for $0\leq j \leq N$. Thus, $h_j$'s are linear in the $y_i$'s, and there exists a matrix $A$ with $m$ columns and coefficients in $\Q[x_1,\ldots,x_n]$~such~that %for all $j$: %$1\leq j \leq n$,
%    \todoremi{Pb notation}
    \begin{equation}\label{eqn:matrix}
    \begin{bmatrix}
        h_1\cdots
        h_m 
    \end{bmatrix}^t
    =
    A \cdot 
    \begin{bmatrix}
        y_1\cdots
        y_m   
    \end{bmatrix}^t
    \end{equation}
    Let $\bb \in \Q^n$, by Proposition~\ref{prop3.4}, $g(\xb,\bb)$ is a polynomial invariant of $\mathcal{L}(\ab,0,F)$ if, and only if, $(\ab,\bb)\in S_{(F_m,X)}$, that is, by \eqref{eqn:matrix}, if and only if $A(a_1,\dotsc,a_n)\cdot \bb = 0$. Since any polynomial in $\Q[x_1,\ldots,x_n]$ can be written as $g(\xb,\bb)$, for some $\bb \in \Q^m$, we are done.
\end{proof}

\begin{remark}\label{remark_disequality}
We can add disequalities of form $p(x) \neq 0$ in the guard loop as in \cite{muller2004computing}. 
Indeed, by applying Algorithm~\ref{algo1} to $(p\cdot g,F_m)$ instead of $(g,F_m)$ in the above proof,  
this implies that at each iteration, either $p(x)=0$ and the loop terminates (and so does Algorithm~\ref{algo1}) or we add the usual constraints given by the polynomial map.
\end{remark}
%Example~\ref{ex2} is taken from~\cite[Page 2]{hrushovski2018polynomial}, where they compute polynomial invariants for specific initial values. In~\cite{hrushovski2018polynomial} , they use polynomial invariants of Example~\ref{ex2} to verify that the linear loop with an assignment $``2y-x\geq -2"$ never terminates. In our case, we compute all polynomial invariants up to degree 2 for every initial value. 

\vspace{-3mm}
\begin{example}\label{ex2} 
We consider the following loop $\mathcal{L}$ from \cite{hrushovski2018polynomial}.

\programbox[0.55\linewidth]{
\State$(x_1, x_2)=(a_1,a_2)$
\While{true}
\State $\begin{pmatrix}
x_1 \\
x_2
\end{pmatrix}
\longleftarrow
\begin{pmatrix}
10x_1-8x_2\\
6x_1-4x_2
\end{pmatrix}
$
\EndWhile
}

We proceed to compute the second truncated polynomial ideal for $\mathcal{L}$ using the algorithm outlined in the proof of Theorem~\ref{algogeneral}. Some of the polynomial invariants for this loop have been computed in \cite{hrushovski2018polynomial} for specific initial values, and are used to verify the non-termination of the linear loop with the assignment $``2x_2-x_1\geq -2"$. In our analysis, we extend this validation by computing \emph{all} polynomial invariants up to degree 2 for \emph{arbitrary} initial value.

We first run Algorithm~\ref{algo1} on input $F_6=(10x_1-8x_2,6x_1-4x_2,y_1,\ldots,y_6)$, and 
$g = y_1+y_2x_1+y_3x_2+y_4x_1^2+y_5x_1x_2+y_6x_2^2$ where the $y_i$'s are new variables. The output is 
% computes the invariant set $S_{(F_{6},X)}$
% assume that $g = b_1+b_2x+b_3y+b_4x^2+b_5xy+b_6y^2$ is a polynomial invariant. % for $\mathcal{L}$. 
% Let $X=\V(b_1+b_2x+b_3y+b_4x^2+b_5xy+b_6y^2)$. %\subset \mathbb{C}^8$. 
%  and outputs
%The primary decomposition of $I(S_{(F_{6},X)})$ \textcolor{blue}{why is that called $F_6$?}can be computed using the following command in Macaulay2~\cite{M2}.
%\[\mathtt{primaryDecomposition}(I(S_{(F_{6},X)})),\]
polynomials $h_1,\dotsc,h_5$ in $\Q[x_1,x_2,y_1,\ldots,y_6]$
%of degree $\leq3$, %less than or equal to $3$, 
whose common zero set is $S_{(F_{6},X)} \subset \C^8$. 

As the $h_i$'s are linear in the $y_j$'s, we can write them as the product of the matrix
% $h_1=y_1$ and $h_2,\ldots,h_4$ as the entries of: 
%the product of the matrices:
%We used ``\texttt{trim}" command from Macaulay~\cite{M2} to find minimal generators for the ideal generated by $g_1,\ldots,g_5$. These generators can be written as the product of the matrix
%\todoremi[inline]{It doesn't look like the output of Algorithm~\ref{algo1}. (e.g. first row?) Yes, the output of Algorithm 1 is very huge that is why I trimed the ideal.}
\[ \scriptsize{\arraycolsep=0.5\arraycolsep
\begin{bmatrix}
        1 & 0 & 0 & 0&  0& 0\\
        0 & 3x_1-4x_2 & 3x_1-4x_2 & 0&
         0& 0\\
        0& 64x_2& 112x_2-48x_1& 48x_2^2&
         84x_2^2-36x_1x_2&27x_1^2-126x_1x_2+147x_2^2\\
         0& 32x_2 & 56x_2-24x_1 & 24x_1x_2& -9x_1^2+21x_1x_2+12x_2^2& -18x_1x_2+42x_2^2\\
         0& 4x_2& 7x_2-3x_1& 3x_1^2&3x_1x_2 &3x_2^2
   \end{bmatrix}
     %  =0
      }\]
by the vector whose entries are the $y_1,\ldots,y_6$. This matrix is the output of the procedure \textsf{TruncatedInvariant} given in Theorem~\ref{algogeneral}.
Here, we actually show a reduced version of this matrix for clarity reasons: we used the ``\texttt{trim}'' command from Macaulay~\cite{M2}, to find smaller generators for the ideal generated by the $g_i$'s. 

Actually, from this output we can go further by computing an explicit basis for the corresponding vector space of $I_{2,\mathcal{L}}$. This is done by computing a basis for the kernel of the above matrix, depending of the possible values for $(a_1,a_2)$.
Performing Gaussian elimination on this matrix, we are led to consider the following four cases:
% using Proposition~\ref{prop3.4}, we apply ``\texttt{primaryDecomposition}'' command from Macaulay2~\cite{M2} to compute a primary decomposition of 
% the ideal $(g_1,\dotsc,g_5)$. %Using the  ``\texttt{primaryDecomposition}'' command from Macaulay2~\cite{M2}.
% This leads to distinguish four cases, depending on the initial values $(a_1,a_2)$: %, that correspond to the four irreducible components of $S_{(F_{6},X)}$.
\begin{center}
\scalebox{0.9}{\begin{tabular}{ |c|c| } 
\hline
Initial values & Basis of $I_{2,\mathcal{L}}$\\
\hline
$a_1=a_2=0$ & $\{x_1,x_2,x_1x_2,x_1^2,x_2^2\}$ \\ \hline
$a_1=a_2\neq 0$ &  $\{x_1-x_2, x_1^2-x_1x_2, -x_1x_2+x_2^2\}$  \\ \hline
$a_1=\frac{4}{3}a_2\neq 0$ & $\{3x_1-4x_2, -3x_1^2+16x_1x_2-16x_2^2, -3x_1x_2+4x_2^2\}$\\ \hline
$a_1\neq\frac{4}{3}a_2$, & $\{(3a_{1}-4a_{2})^{2}x_1-(3a_{1}-4a_{2})^{2}x_2-9(a_1-a_2)x_1^2$\\ 
$a_1 \neq a_2$ & $+24(a_1-a_2)x_1x_2-16(a_1-a_2)x_2^2\}$
\\
\hline
\end{tabular}}
\end{center}
\begin{comment}
    [0.5em]
\noindent \underline{Case 1:}  $a_1=a_2=0$. Then, 
    % $f=0$ is a polynomial invariant if and only if $a=0$. Therefore,
    $\{x,y,xy,x^2,y^2\}$ is a basis for $I_{2,\mathcal{L}}.$
    % \text{ and } \langle I_{2, \mathcal{L}}\rangle = \langle x, y \rangle$$

\vspace*{0.5em}
\noindent\underline{Case 2:} $a_1=a_2\neq 0$. Then,
        % $f=0$ is a polynomial invariant if and only if $a=0,$ $ d+e+w=0$ and $b+c=0$. Therefore,
    $\{x-y, x^2-xy, -xy+y^2\}$ is a basis for $I_{2,\mathcal{L}}.$
    % \text{ and } \langle I_{2, \mathcal{L}}\rangle = \langle x- y \rangle$$

\vspace*{0.5em}
\noindent \underline{Case 3:} $3a_1 = 4a_2\neq 0$. Then, $\{3x-4y, -3x^2+16xy-16y^2, -3xy+4y^2\}$ is a basis for $I_{2,\mathcal{L}}.$
        % $f=0$ is a polynomial invariant if and only if $a=0,$ $ 12d+16e+9w=0$ and $4b+3c=0$. Therefore,
         % \text{ and } \langle I_{2, \mathcal{L}}\rangle = \langle 3x- 4y \rangle$$

\noindent \underline{Case 4:} $3a_1 \neq 4a_2$ and $a_1 \neq a_2$. Then,
        % $f=0$ is a polynomial invariant if and only if $a=0, \, 16e=9w, \,  2d=-3w, \, w(3x_0-4y_0)^2=16c(x_0-y_0)$ and $b+c=0$. Therefore, 
         $\{g\}$ is a basis for $I_{2,\mathcal{L}}$, 
% and $\langle I_{2, \mathcal{L}}\rangle = \langle f\rangle$, 
where 
\[
\begin{array}{ll}
    g= &(3a_{1}-4a_{2})^{2}x-(3a_{1}-4a_{2})^{2}y-9(a_1-a_2)x^2\\&+24(a_1-a_2)xy-16(a_1-a_2)y^2.
\end{array}
\]
\end{comment}

It is remarkable that in the first three cases, the truncated invariant ideal does not depend on the initial value. This is because they correspond to degenerate cases where the initial values are not generic; that is, they lie in a proper algebraic variety of $\mathbb{C}^2$. However, the last case is generic, and the output depends on the initial values. This is the output one would obtain by running Gauss elimination on the above polynomial matrix in the field of rational fractions in the $x_i$'s. However, such a computation is not tractable in general, as the size of the expressions increases quickly.
\end{example}

% \begin{remark}\todo{Expand?}
% In Example~\ref{ex2}, we note that our Algorithm~\ref{algo1} produces a comprehensive polynomial invariant within the last case, wherein each coefficient of the polynomial depends on the initial values.
% \end{remark}

%\todoremi[inline]{From here, change the variables $x,y$ and z in the examples with $x_1,x_2$ and $x_3$ to avoid confusion}

\subsection{Loops with given initial value}
%Algorithm~\ref{algo1} fails to terminate within a one-hour timeframe for the degree $2$ invariant polynomials. Thus, prompting the need for an algorithm considering loop initial values, and it fails to terminate within a one-hour timeframe for the degree $2$ invariant polynomials. 
%As demonstrated in Example~\ref{ex2}, Algorithm~\ref{algo1} is not reliant on the initial values of a loop. Thus, prompting the need for an algorithm considering loop initial values. %However, in the majority of cases, loops are provided with initial values. 
While the algorithm outlined in Theorem~\ref{algogeneral} addresses the most general case, in practice, it quickly becomes impractical, even for small inputs. In this section, we focus on the particular case where the initial values of the loops are fixed and design an adapted algorithm that is more efficient for this scenario. We will see in Section~\ref{sec:lift} that the solution to this specific problem can be used to partially solve the general problem.
% to enhance the computational efficiency of Algorithm~\ref{algo1}. 

The following proposition provides a sufficient condition for a polynomial to be an invariant, using the loop's initial values.

%\textcolor{blue}{We now present a sufficient condition for a polynomial to be an invariant polynomial of a loop, in the case that an initial value of a loop is given.}
\begin{proposition}\label{prop_sufficient}
Consider a loop $\mathcal{L}(\mathbf{a_0},0,F)$. Let $\mathbf{a_n} = F^{(n)}(\mathbf{a_0})$. If $\displaystyle\sum_{i=1}^{m} y_{i}{\bf x}^{\alpha_{i}}$ is a polynomial invariant, then the $y_i$'s~satisfy~the~equations:
  \vspace{-2mm}
\begin{equation*}\label{le}
    \sum_{i=1}^{m} y_{i}{\bf a_0}^{\alpha_{i}} = \cdots =
    \sum_{i=1}^{m} y_{i}{\bf a_k}^{\alpha_{i}}=0.
\end{equation*}
%where $\mathbf{a_n} = F^{(n)}(\mathbf{a_0})$. 
\end{proposition}
Proposition~\ref{prop_sufficient} is a direct consequence of~the~following~lemma.

\begin{lemma}\label{prop4}
Let $\ab_0 \in \C^n$ and $F = (f_1,\dotsc,f_n) \subset \C[x_1,\dotsc,x_n]$.
Let $$X = \V(\displaystyle\sum_{i=1}^{m} y_{i}{\bf x}^{\alpha_{i}})\subset \mathbb{C}^{n+m}.   \vspace{-2mm}
$$ 
%\vspace*{-0.1cm}
Let 
$X_{k} = \displaystyle\bigcap_{j=0}^{k}F_{m}^{-j}(X)$, 
$S_{k}=X_{k}\cap \V(\xb-\ab_0)$, and $\mathbf{a_k} = F^{(k)}(\mathbf{a_0})$ for all $k\in \N$.
%, land let
%$$
%X_{k} = \displaystyle\bigcap_{j=0}^{k}F_{m}^{-j}(X) 
%\text{\quad and \quad}  S_{k}=X_{k}\cap \V(\xb-\ab_0)%_{1}-a_{0,1},\ldots, x_{n}-a_{0,n})
%.$$
Then, the following holds:
    \begin{itemize}
        \item[$(a)$] $S_k=\V(\displaystyle\sum_{i=1}^{m} y_{i}{\bf a_0}^{\alpha_{i}}, \ldots, \displaystyle\sum_{i=1}^{m} y_{i}{\bf a_{k}}^{\alpha_{i}}, \;\xb - \ab_0).
        % x_{1}-a_{0,1},\ldots, x_{n}-a_{0,n}).
        $\\[0em]
        \item[$(b)$] 
    $S_{(F_{m},X)}\cap \V(\xb-\ab_0)\subset S_{k} \text{ for any } k\in \N$.   
    \end{itemize}
\end{lemma}
\begin{proof}
 $(a)$  Since $F_m=(f_1,\dotsc,f_n,y_1,\dotsc,y_m)$ then for $j\geq 0$, we can note $F_m^{(j)}=(f_{j,1},\ldots , f_{j, n},y_1,\dotsc,y_m)$. Then, according to Lemma~\ref{lem:equtionspreimage}, we can rewrite $X_k$ as\\[-0.3em]
 % $X_k$ is the intersection for all $0\leq j \leq k$ of the sets
   $$
   % \begin{array}{ll}
   \bigcap_{j=0}^{k} F_{m}^{-j}\left(\V(\sum_{i=1}^{m} y_{i}{\bf x}^{\alpha_{i}})\right)
   % &=V(\displaystyle\sum_{i=1}^{m} f_{j,n+i}({\bf x,y})(f_{j,1}({\bf x,y}),\ldots, f_{j,n}({\bf x,y}))^{\alpha_{i}})\\
   \; = \;
   \V\left(\sum_{i=1}^{m} y_{i}{\bf x}^{\alpha_{i}},\ldots,
   \sum_{i=1}^{m} y_{i}(F^{k}({\bf x}))^{\alpha_{i}}\right),\vspace*{0.7em}
   % \end{array}
   $$
   %\todoremi{over-technical}
  so that  
  \vspace{-4mm}
  % is the intersection for all $0\leq j \leq k$ of the sets
 \begin{align*}
     S_k &= \V(\sum_{i=1}^{m} y_{i}{\bf x}^{\alpha_{i}},\ldots, \sum_{i=1}^{m} y_{i}(F^{k}({\bf x}))^{\alpha_{i}},\;\xb - \ab_0)\\
     &= \V(\sum_{i=1}^{m} y_{i}{\bf a_0}^{\alpha_{i}}, \ldots, \sum_{i=1}^{m} y_{i}{\bf a_{k}}^{\alpha_{i}},
     \;\xb - \ab_0).
 \end{align*}
 % \begin{equation*}
 %    \begin{split}
 %    % S_k & =  
 %  %  & %=\Big(\displaystyle\bigcap_{j=0}^{k}V(\displaystyle\sum_{i=1}^{m} y_{i}(F^{j}({\bf x}))^{\alpha_{i}})\Big)\cap V(x_{1}-a_{0,1},\ldots, x_{n}-a_{0,n})\\
 %    V(\sum_{i=1}^{m} y_{i}(F^{j}({\bf x}))^{\alpha_{i}},x_{1}-a_{0,1},\ldots, x_{n}-a_{0,n})\\
 %  = V(\sum_{i=1}^{m} y_{i}(F^{j}({\bf a_0}))^{\alpha_{i}},x_{1}-a_{0,1},\ldots, x_{n}-a_{0,n})\\
 %    \end{split}
 %    \end{equation*}   

 %    That is 
 %    \[
    
 %    \]

 $(b)$ By Proposition~\ref{prop:stabilization}, we have the following descending chain:
   $$X_{0}\supset X_{1}\supset \ldots \supset X_{N}= S_{(F_{m},X)}=X_{N+1}= \ldots \text{ for some } N\in \N.$$
   Thus, by intersecting the above chain with an algebraic variety $V=\V(x_{1}-a_{0,1},\ldots, x_{n}-a_{0,n})$ we get the descending chain:
    $$S_{0}\supset S_{1}\supset \ldots \supset S_{N}= S_{(F_{m},X)}\cap V %V(x_{1}-a_{0,1},\ldots, x_{n}a_{0,n})
    =S_{N+1}= \ldots \text{ for some } N\in \N.$$
   % for some $N\in \N$.
    Thus, $S_{(F_{m},X)}\cap V$ %V(x_{1}-a_{0,1},\ldots, x_{n}-a_{0,n})$ 
    is a subset of $S_{k}$ for any $k\in \N$.
\end{proof}

Therefore, when a loop's initial value is set, Proposition~\ref{prop_sufficient} can provide as many $M$ linear equations as desired, for the coefficients of a degree $d$ polynomial invariant.
Since the codimension of the dth truncated ideal is bounded by $\binom{n+d}{d}$, the dimension of the vector space of polynomials of degree $\leq d$ in $x_1,\dotsc,x_n$. Hence, this latter bound is a natural choice for $M$ as it corresponds to the optimal number of equations in the case where no non-trivial invariant exist.
Indeed, these equations, accelerate the polynomial invariant computations by providing a superset of the desired truncated invariant ideal.
In particular, a vector basis $\mathcal{B}$ of this solution set serves as candidates for the basis of the truncated invariant ideal. 

Following the above idea, we present an algorithm to compute a vector basis of the dth truncated invariant ideal of a loop with a fixed initial value. We first explain two subroutines used in this algorithm:
(i) The procedure \textsf{VectorBasis} takes linear forms as input and computes a vector basis of the common vanishing set of these forms.
(ii) The procedure \textsf{CheckPI} takes as input ${\bf a} \in \Q^n$ and polynomials $F=(f_1,\ldots,f_n)$ and $g$ in $\Q[x_1,\ldots,x_n]$. It outputs \texttt{True} if, and only, if $g$ is a polynomial invariant of $\mathcal{L}(\ab,0,F)$. Such a procedure can be obtained by a direct combination of an application of Algorithm~\ref{algo1} to $(g,F)$ and the effective criterion given by Proposition~\ref{prop:terminateinvariant}.

% It returns ``True'' if $(h_1(a), \ldots, h_k(a)) =0$ where $(h_1, \ldots, h_k) \gets \text{InvariantSet}(g, F)$, and ``False'' otherwise.
% The procedure ``CheckPI'' takes a sequence of rational numbers ${\bf a}$, a sequence of polynomials $F$, and a polynomial $g$ as input. It returns ``True'' if $(h_1(a), \ldots, h_k(a)) =0$ where $(h_1, \ldots, h_k) \gets \text{InvariantSet}(g, F)$, and ``False'' otherwise. Consequently, $\CheckPI({\bf a}, F, g)$ assesses whether $g=0$ constitutes a polynomial invariant for a loop with initial value ${\bf a}$ and a polynomial map $F$ based on Proposition~\ref{prop_sufficient}. A common choice for the parameter $M$ is $\binom{n+d}{d}$ as outlined in Algorithm~\ref{alg2}, although it can be chosen differently based on the underlying loop.

\begin{algorithm}[H]
\caption{Computing truncated invariant ideals}\label{alg2}
\begin{algorithmic}[1]\setstretch{1.2}
\Require A sequence of rational numbers ${\bf a} = (a_1,\ldots,a_n)$, a natural number $d$ and %a sequence of 
polynomials $F=(f_1,\ldots,f_n) \in \mathbb{Q}[x_1,\ldots,x_n]$.
\Ensure Polynomials forming a vector space basis for the dth truncated ideal of the loop $\mathcal{L}(\ab,0,F)$. % with initial value $\bf a$ and a polynomial map $F$. $\mathcal{L}(\ab,F)$?}
\State $ g\gets\sum\limits_{|\alpha_i|\leq d}y_i{\bf x}^{\alpha_i};$
\State$M\gets \binom{n+d}{d};$
\State $(b_1,\ldots,b_m)\gets \VectorBasis\big(g({\bf a,y}),g(F({\bf a),y}), \ldots g(F^{M}({\bf a),y})\big)$;
\State \label{step:B}$\mathcal{B}\gets \big(\sum\limits_{|\alpha_i|\leq d}b_{1,i}{\bf x}^{\alpha_i},\ldots ,\sum\limits_{|i|\leq d}b_{m,i}{\bf x}^{\alpha_i}\big)$;
\State$C=(h_1,\ldots,h_l) \gets \{ h\in \mathcal{B}\mid \CheckPI({\bf a},F,h)==False\};$
\If{$C ==\emptyset$,}
%\algstore{testcont}

%\end{algorithmic}
%\end{algorithm}  
%\begin{algorithm}
%\begin{algorithmic}[1]
%\algrestore{testcont}

\State \Return $\mathcal{B};$
\Else
\State$(\widetilde{h}_1,\ldots, \widetilde{h}_k) \gets \InvariantSet(\sum_{j=1}\limits^{l} z_jh_j, (f_1,\ldots , f_n, z_1,\ldots, z_l ));$
\State$(b_{1}',\ldots, b_{s}')\gets \VectorBasis(\widetilde{h}_1({\bf a,z}),\ldots, \widetilde{h}_k({\bf a,z}))$;
\State$\mathcal{B}_1\gets \big(\sum\limits_{j=1}^{l}b_{1,i}'h_j, \ldots, \sum\limits_{j=1}^{l}b_{s,i}'h_j\big)$;
\State$\mathcal{B}_2\gets \mathcal{B}.remove(C); $
\State$\mathcal{B}\gets\mathcal{B}_1.extend(\mathcal{B}_2);$
\State \Return $\mathcal{B};$
\EndIf
\end{algorithmic}
\end{algorithm}

\begin{remark}\label{remark_disequality_fixed}
As described in Remark~\ref{remark_disequality}, for the general case, one can easily adapt this algorithm to handle loops with disequalities $p(x)\neq 0$ in the guard by replacing $g$ with $p\cdot g$ at step 3.

In terms of complexity, in the worst case, Algorithm~\ref{alg2} does not find any candidates at step 3 and calls Algorithm~\ref{algo1} on the general loop at step 9, without exploiting the given initial values. However, in practice, all candidates found at step 3 are invariants (see Section~\ref{sec:implementation}), and Algorithm~\ref{alg2} terminates at step 7.
\end{remark}

%\todoremi[inline]{I don't understand the else block of Algorithm 2}

%\begin{remark}Algorithm~\ref{alg2} can be readily extended to compute all polynomial invariants of a specified form, particularly when the potential non-zero terms are predefined, as \textcolro{red}{illustrated in Example~\ref{ex12}.}\end{remark}

We now prove the correctness of Algorithm~\ref{alg2}.

\begin{theorem}\label{thm:corralgo2}
    On input a sequence of ${\bf a} = (a_1,\ldots,a_n)$ in $\mathbb{Q}^n$, a sequence of polynomials $F=(f_1,\ldots,f_n) \in \mathbb{Q}[x_1,\ldots,x_n]$ and $d\in \N$, Algorithm~\ref{alg2} outputs a sequence of polynomials which is a basis for the dth truncated ideal for the loop $\mathcal{L}({\bf a},0,F)$. %a loop $\mathcal{L}$ with initial value $\bf a$ and a polynomial map $F$.
\end{theorem}
\begin{proof}
Assume that $ g=\sum\limits_{|\alpha_i|\leq d}y_i{\bf x}^{\alpha_i}$ is a polynomial invariant for $\mathcal{L}$.
Let $\{b_1,\ldots, b_m\}$ be a basis for the solution set of $g({\bf a,y})=g(F({\bf a),y})= \cdots =g(F^{M}({\bf a),y})=0$ where $M=\binom{n+d}{d}$. Then, 
$$\mathcal{B}=\{\sum\limits_{|\alpha_i|\leq d}b_{1,i}{\bf x}^{\alpha_i},\ldots ,\sum\limits_{|i|\leq d}b_{m,i}{\bf x}^{\alpha_i}\}$$ 
consists of linearly independent polynomials in $\C[x_1,\ldots,x_n]_{\leq d}$.
By Proposition~\ref{prop_sufficient}, the variables $\bf y$ satisfy linear equations $g({\bf a,y})=g(F({\bf a),y})= \cdots =g(F^{M}({\bf a),y})=0$. Therefore, $I_{d,\mathcal{L}}$ is contained in the vector space generated by $\mathcal{B}$.  Let $C=\{h_1,\ldots,h_l\}$ be the set of all polynomials in $\mathcal{B}$ that are not polynomial invariants. Assume that $C$ is not empty. Otherwise, every polynomial in $\mathcal{B}$ is a polynomial invariant for $\mathcal{L}$ which implies that $\mathcal{B}$ is a basis for $I_{d,\mathcal{L}}$. By Proposition~\ref{prop3.4}, $\sum\limits_{j=1}^{l} z_jh_j$ is a polynomial invariant if and only if $\widetilde{h}_1({\bf a,z})=\cdots =\widetilde{h}_k({\bf a,z})=0$, where  $$(\widetilde{h}_1,\ldots, \widetilde{h}_k) = \InvariantSet(\sum\limits_{j=1}^{l} z_jh_j, (f_1,\ldots , f_n, z_1,\ldots, z_l )).$$
Since $\widetilde{h}_1({\bf a,z})=\cdots =\widetilde{h}_k({\bf a,z})=0$ represents a system of linear equations, we can find a basis $\mathcal{B}_1$ for a subspace $V_1$ that is the intersection of $I_{d,\mathcal{L}}$ and the vector space generated by $C$, using exactly the same method employed for computing $\mathcal{B}$. 
Denote the set $\mathcal{B}\setminus C=\{g_{1},\ldots,g_{m-l}\}$ by $\mathcal{B}_2$.  

Now, we will show that the vector space $V$ generated by $\mathcal{B}_1 \cup \mathcal{B}_2$ is equal to $I_{d,\mathcal{L}}$. Since $\mathcal{B}_1 \cup \mathcal{B}_2 \subseteq I_{d,\mathcal{L}}$, it follows that $V$ is a subset of $I_{d,\mathcal{L}}$. To prove the other inclusion, let $g\in I_{d,\mathcal{L}}$. Since $I_{d,\mathcal{L}}$ is contained in the vector space generated by $\mathcal{B}$, there exist $c_1,\ldots,c_m\in \C$ such that $g=\sum\limits_{i=1}^{l}c_i\widetilde{h}_i+\sum\limits_{j=1}^{m-l}c_{l+j}g_j$. Then, $g-\sum\limits_{j=1}^{m-l}c_{l+j}g_j$ is a polynomial invariant given that $g,g_{l+1},\ldots,g_m$ are all polynomial invariants. Thus, $\sum\limits_{i=1}^{l}c_ig_i$ is a polynomial invariant contained in the vector space generated by $C$, implying $\sum\limits_{i=1}^{l}c_ig_i\in V'$. Hence, $g\in V$ and so, $I_{d,\mathcal{L}}\subseteq V$. Since the intersection of the vector space generated by $\mathcal{B}_1$ and the vector space generated by $\mathcal{B}_2$ is $\{0\}$, we conclude that $\mathcal{B}_1\cup \mathcal{B}_2$ is a basis for $I_{d,\mathcal{L}}$, which completes the proof. %proves the correctness of Algorithm~\ref{alg2}.
\end{proof}  
%  \textcolor{blue}{Additionally, it is possible to generate general polynomial invariants of the form $f({\bf x})-f({\bf x_0})=0$ using Algorithm~\ref{algo1} and polynomial invariants that are computed by Algorithm~\ref{alg2} (see Examples~\ref{ex5},~\ref{ex6},~\ref{ex7},~\ref{ex8},~\ref{ex9}).}
%\noindent We now use Algorithm~\ref{alg2} for computing polynomial invariants in an example.% in Example~\ref{ex3}. % up to degree 4. % using Algorithm~\ref{alg2}. 

\begin{example}[Squares]\label{ex3}
Consider the ``Squares'' loop in Appendix~\ref{bench}.
For $d=2,3,4$, the algorithm \textsf{TruncatedInvariant} 
given by Theorem~\ref{algogeneral} cannot compute the $d$-th truncated ideal, %, as in Example~\ref{ex2}, 
within an hour.
However, when initial values are fixed, Algorithm~\ref{alg2} easily computes them. To compute $I_{2,\mathcal{L}}$, the input for Algorithm~\ref{alg2} is $(-1,-1,1)$, $2,$ and $F$. Assume that 
$g=y_1+y_2x_1+\cdots+y_{10}x_3^2$ 
%$      g= y_1+y_2x_1+y_3x_2+y_4x_3+y_5x_1^2+y_6x_1x_2 
 %       +b_{7}x_1x_3+b_{8}x_2^2+b_{9}x_2x_3+b_{10}x_3^2$ 
is a polynomial invariant. % of degree $2$. 
By Proposition~\ref{prop_sufficient}, this leads to 10 linear equations, whose solutions 
 \begin{comment}
\[ \begin{bmatrix}
         1 & -1 & -1 & 1& 
         1& 1&-1 &1 &-1&1\\
         1 & 1 & -3 & 1&
         1& -3& 1 &9 &-3&1\\
         1& 12& -13& 0&
         144&-156&0&169&0&0\\
         \vdots & \vdots & \vdots & \vdots& \vdots& \vdots& \vdots& \vdots& \vdots& \vdots
   \end{bmatrix}
     \begin{bmatrix}
         b_{2} \\
         b_{1}\\
         b_{3}\\
         \vdots 
     \end{bmatrix}
      =
     0\]
Most of the remaining entries in the above matrix have more than 5 digits, and the longest entries have 147 digits. Using a basis for the solution set of the above matrix, we compute  
\end{comment}
give the following candidates for polynomial invariants: 
\[
  \begin{array}{lll}
  \mathcal{B}=     &  \{1+x_1+x_2+x_3, 1+x_1+x_2+x_3^2,  2+3x_1+3x_2\\
       &+x_1^2+2x_1x_2+x_2^2,
        -2-x_1-3x_2+x_1^2+2x_1x_3-x_2^2,\\
        &2-3x_1-x_2-x_1^2+x_2^2+2x_2x_3\}.
  \end{array}
  \]
The procedure $\CheckPI$ verifies that all the polynomials in $\mathcal{B}$ are invariant polynomials, and $\mathcal{B}$ forms a basis for $I_{2,\mathcal{L}}$. Moreover, the outputs of Algorithm~\ref{alg2} 
%for $I_{3,\mathcal{L}}$ and $I_{4,\mathcal{L}}$ 
show that %the truncated ideal 
$I_{3,\mathcal{L}}$ represents a 13-dimensional vector space, and $I_{4,\mathcal{L}}$ a 26-dimensional vector space.

This example has been previously explored in \cite{Unsolvableloops}, where only a closed formula is derived as $x_1(n)+x_2(n)=2^n(x_1(0)+x_2(0)+2)-(-1)^n/2-3/2$. In particular, they do not find any of the above invariants. 
We calculate truncated invariant ideals $I_{d, \mathcal{L}}$ for $d=1,2,3,4$, considering a given initial value. This covers all polynomial invariants up to degree 4 for Example~\ref{ex3} using Algorithm~\ref{alg2}. 
\end{example}

%%%%%%%%%%%%%%%%%%%%%%%%%%%%%%%%%%%%%%%%%%%%%%%%%%%%%%%%%%
%%%%%%%%%%%%%%%%%%%%%%%%%%%%%%%%%%%%%%%%%%%%%%%%%%%%%%%%%%
\section{Applications and further results}\label{sec:application}
%%%%%%%%%%%%%%%%%%%%%%%%%%%%%%%%%%%%%%%%%%%%%%%%%%%%%%%%%%
%%%%%%%%%%%%%%%%%%%%%%%%%%%%%%%%%%%%%%%%%%%%%%%%%%%%%%%%%%

In this section, we show different applications of Algorithm~\ref{alg2} and their consequences to various examples from the literature. %In particular, we compute \textcolor{red}{what is the contribution compared to the aforementioned literature}

%%%%%%%%%%%%%%%%%%%%%%%%%%%%%%%%%%%%%%%%%%%%%%%%%%%%%%%%%%
\subsection{On the (generalized) Fibonacci sequences}
%%%%%%%%%%%%%%%%%%%%%%%%%%%%%%%%%%%%%%%%%%%%%%%%%%%%%%%%%%

The following examples, discussed in \cite{Unsolvableloops}, are significant in the theory of trace maps, see e.g.~\cite{baake1993trace, roberts1994trace}. In each example, Algorithm~\ref{alg2} computes truncated polynomial ideals up to degree 4, establishing that in each case, there are no polynomial invariants up to degree~2. 

In Example~\ref{fibex}, we prove that the computed polynomial invariant of degree 4 by Algorithm~\ref{alg2}, generates the entire invariant ideal. Note that the polynomial invariants of the Fibonacci loop, and more generally linear loops, are efficiently addressed by \cite{hrushovski2018polynomial,hrushovski2023strongest}.

%In the following example, we compute the 4th truncated polynomial invariant for Fibonacci sequence. Moreover, we prove that the computed polynomial invariant generates the whole invariant ideal.

%\textcolor{blue}{Uniqueness of polynomial invariant}
\begin{example}[Fibonacci sequence]\label{fibex}
The Fibonacci numbers follow the recurrence relation: $F_{0}=0, F_{1} = 1$, and $F_{n} = F_{n-1}+F_{n-2}$ for all $n\geq 2$. We can express the Fibonacci sequence as a loop $\mathcal{F}$.

\programbox[0.55\linewidth]{
\State$(x_1, x_2)=(0,1)$
\While{true}
\State $\begin{pmatrix}
x_1 \\
x_2
\end{pmatrix}
\longleftarrow
\begin{pmatrix}
x_2\\
x_1+x_2
\end{pmatrix}
$
\EndWhile
}

\noindent Algorithm~\ref{alg2} computes that the truncated invariant ideals $I_{2, \mathcal{F}}$ and $I_{3, \mathcal{F}}$ are zero, and $g={-1+x_1^4+2x_1^3x_2-x_1^2x_2^2-2x_1x_2^3+x_2^4}$ forms a basis for $I_{4, \mathcal{F}}$. Therefore, the Fibonacci numbers satisfy the equation:
\[
F_{n-1}^4+2F_{n-1}^3F_{n}-F_{n-1}^2F_{n}^2-2F_{n-1}F_{n}^3+F_{n}^4-1=0 \text{ for all } n\in \mathbb{N}.
\]
Moreover, $I_{\mathcal{F}}$ is generated by $g$. To prove that, observe that
$$g=(-1-x_1^2-x_1x_2+x_2^2)(1-x_1^2-x_1x_2+x_2^2)$$ 
and $(F_{n-1},F_{n})$ lies on $\V(1-x_1^2-x_1x_2+x_2^2)$ for infinitely many $n$. 
Thus, for any $f\in I_{\mathcal{F}}$, the system of equations $1-x_1^2-x_1x_2+x_2^2=f(x_1,x_2)=0$ has infinitely many solutions. 
By~\cite[Page 4]{shafarevich1994basic}, $f(x_1,x_2)$ is divisible by $1-x_1^2-x_1x_2+x_2^2$. The same applies to $-1-x_1^2-x_1x_2+x_2^2$. Consequently, $f(x_1,x_2)$ is divisible by $g$, establishing that $I_{\mathcal{F}}$ is generated by $g$. Note that $g$ can be easily derived from the square of Cassini's identity, see e.g.~\cite{kauers2008computing}.
\end{example}

For the following examples, Algorithm~\ref{alg2} computes a unique~%polynomial 
invariant in degree 3. Proposition~\ref{prop3.4} and Algorithm~\ref{alg2} demonstrate that the identified polynomial is the sole invariant of degree~3. %, and all degree 4 polynomials can be derived from it. 
The uniqueness proof is a novel contribution. Additionally, the polynomials given in \cite{amrollahi2023solvable} for Fib2 and Fib3 %~\ref{ex6}-\ref{ex7} 
were found to be incorrect.

\begin{example}%[fib1, fib2, fib3]
\label{fib1-3} %Let us compute truncated invariant ideals $I_{2,\mathcal{L}}$, $I_{3,\mathcal{L}}$ and $I_{4,\mathcal{L}}$ for the 
For the Fib1, Fib2 and Fib3 loops from Appendix~\ref{bench} %following loop $\mathcal{L}$.
\begin{comment}
\programbox[0.55\linewidth]{
\State$(x_1, x_2,x_3)=(2,1,1)$
\While{true}
\State $\begin{pmatrix}
x_1 \\
x_2 \\
x_3
\end{pmatrix}
\longleftarrow
\begin{pmatrix}
x_2\\
x_3\\
2x_2x_3-x_1
\end{pmatrix}
$
\EndWhile
}
\end{comment}
Algorithm~\ref{alg2} computes a basis for the truncated invariant ideals as:
\[
I_{1,\mathcal{L}}= I_{2, \mathcal{L}} = \{0\},\quad I_{3,\mathcal{L}}=\{g\},\quad\text{and}\quad I_{4,\mathcal{L}}=\{
g,x_1g,x_2g,x_3g
\},
\]
where $g$ for Fib1, Fib2 and Fib3 is, respectively,
\begin{itemize}
    \item[] $-2+x_1^2+x_2^2+x_3^2-2x_1x_2x_3$, \ 
     $76 -x_2 -2x_1x_3+4x_1^2x_2$, and 
    \item[] $7+x_1+x_2+x_3-x_1^2+x_1x_2+x_1x_3-x_2^2+x_2x_3-x_3^2+x_1x_2x_3$.
\end{itemize} 
Note that a basis for $I_{4,\mathcal{L}}$ is generated by a basis for $I_{3,\mathcal{L}}$.
%  Using Algorithm~\ref{algo1} with input a polynomial map $F_{1}$ and an algebraic variety $X=V(x^2+y^2+z^2-2xyz-2-t)\subset \mathbb{C}^4$, we compute the invariant set $S_{(F_{1},X)}=V(x^2+y^2+z^2-2xyz-2-t)$ which means $x^2+y^2+z^2-2xyz-2-t=0$ is a polynomial invariant if and only if $t=x_{0}^2+y_{0}^2+z_{0}^2-2x_{0}y_{0}z_{0}-2$ where $(x_{0}, y_{0}, z_{0})$ is an initial value of $\mathcal{L}$. Therefore, $x^2+y^2+z^2-2xyz-(x_{0}^2+y_{0}^2+z_{0}^2-2x_{0}y_{0}z_{0})=0$ is a general polynomial invariant.
\end{example}

\subsection{Invariant lifting for generic initial values}\label{sec:lift}
%%%%%%%%%%%%%%%%%%%%%%%%%%%%%%%%%%%%%%%%%%%%%%%%%%%%%%%%%%
In this section, we present a method for deriving a polynomial invariant for any initial value from a specific one. Given a polynomial invariant $f$ for a loop $\mathcal{L}$ with a given initial value, our method checks if $f-f(a)$ is a polynomial invariant for $\mathcal{L}$ for any initial value $a$. Additionally, $f-h(a)$ is a polynomial invariant for $\mathcal{L}$ for any initial value $a$ if and only if $h=f$.

%In this section, we introduce a method that computes a polynomial invariant for any initial value from a polynomial invariant for a specific initial value. Assume that we have a polynomial invariant $f$ for a loop $\mathcal{L}$ with a given initial value. Then, our method checks whether $f-f(a)$ is a polynomial invariant for $\mathcal{L}$ for any initial value $a$. Moreover, $f-h(a)$ is a polynomial invariant for $\mathcal{L}$ for any initial value $a$ if and only if $h=f$. 

\begin{comment}

\begin{example}\label{computgeneral}\todo{remove in case of space issue?}Consider the loop from Example~\ref{exadd}.
We compute general polynomial invariants of the form $f(x,y,z)-f(x_0,y_0,z_0)=0$ for $\mathcal{L}$. Using Algorithm~\ref{algo1} with input $z^2-z-t$ and $F$, we compute the invariant set $S_{(F_{1},X)}=V(z^2-z-t)$, implying that $z^2-z-t=0$ is a polynomial invariant if and only if $z_{0}^2-z_{0}-t=0$. Therefore, $z^2-z-(z_{0}^2-z_0)=0$ is a polynomial invariant for any initial value.
%Using Algorithm~\ref{algo1} with input $1+x+y+z+t$ and $F$, we compute the invariant set $S_{(F_{1},X)}=V(t,1+x+y+z)$, implying that $1+x+y+z-t=0$ is a polynomial invariant if and only if $t=0$ and $1+x_0+y_0+z_0=0$, providing a condition for the initial value. Therefore, $1+x+y+z=0$ cannot be extended to a polynomial invariant for any initial value.

%\textcolor{blue}{Another application: In the second case, the first equation only depends on initial values. Therefore, there is no general polynomial invariant of the form $f(x)-f(x_0)=0$.}
\end{example}
\end{comment}
\begin{proposition}\label{general}
    A polynomial invariant $f({\bf x})=0$ for a loop $\mathcal{L}$ with given initial value and a polynomial map $F$ can be extended to a polynomial invariant $f({\bf x})-f({\bf a})= 0$ for $\mathcal{L}$ with any initial value $\bf a$ if and only if $S_{(F_1, X)}=X$, where $X=\V(f({\bf x})-t)$.
\end{proposition}
\begin{proof}
Assume that $f(\mathbf{x})-f(\mathbf{a})=0$ is %a polynomial 
an invariant for $\mathcal{L}$ for any initial value $\mathbf{a}$. Thus, $\mathcal{L}_{1}$ with a guard $f(x)-t=0$ never terminates if $t =f(\mathbf{a})$. Hence, $X=\V(t-f(\mathbf{x}))\subset S{(F_1, X)}$ and $S_{(F_1, X)} \subset X$ by Proposition~\ref{prop:stabilization}, implying that $X=\V(S_{(F_1,X)})$.
    To prove the other direction, assume $S_{(F_1,X)}= X$. By the definition of the invariant set, $\mathcal{L}_1$ with a guard $f(x)-t=0$ never terminates if and only if initial value of $({\bf x},t)$ is contained in $S_{(F_1,X)}$. Thus, $\mathcal{L}_1$ never terminates if and only if $t= f({\bf a})$ for any initial value $a
    $. By Proposition~\ref{prop3.4}, $f(x)-f({\bf a})=0$ is an invariant for $\mathcal{L}$ with any initial value $\bf a$.
\end{proof}
\vspace{-3mm}

\begin{example}
Consider the loops in Example~\ref{fib1-3}. We compute polynomial invariants of the form $f(x_1,x_2,x_3)-f(a_1,a_2,a_3)=0$. Algorithm~\ref{algo1}, with the input being a polynomial map $F_{1}$ and the algebraic variety $X=\V(x_1^2+x_2^2+x_3^2-2x_1x_2x_3-2-t)\subset \mathbb{C}^4$,  
computes $S_{(F_{1},X)}$, which is equal to $\V(x_1^2+x_2^2+x_3^2-2x_1x_2x_3-2-t)$. 
Then, $x_1^2+x_2^2+x_3^2-2x_1x_2x_3-(a_{1}^2+a_{2}^2+a_{3}^2-2a_{1}a_{2}a_{3})=0$ is a general invariant for Fib1, %, Example~\ref{fib1-3}, 
as stated in Proposition~\ref{general}. Similarly, we verify that polynomial invariants of Fib2 and Fib3 %Example~\ref{ex6} and Example~\ref{ex7} 
can be extended to a general polynomial invariant of the form $f(x_1,x_2,x_3)-f(a_1,a_2,a_3)=0$.
\end{example}

\subsection{%A sufficient criterion for 
Termination of semi-algebraic loops}
%%%%%%%%%%%%%%%%%%%%%%%%%%%%%%%%%%%%%%%%%%%%%%%%%%%%%%%%%%%%%%%%%%%%%%%%%%%

%\todoremi[inline]{Yet to be written}

\begin{definition}\label{def:saloop}
Consider the basic semi-algebraic set $S$ of \, $\R^n$ defined by $g_1=\cdots=g_k=0$ and $h_1>0,\ldots,h_s>0$ and a polynomial map $F= (f_1,\ldots,f_n)$, where the $f_i$'s, the $g_j$'s  and the $h_j$'s are polynomials in $\mathbb{R}[x_1,\ldots,x_n]$.
Then a loop of the form:

\programbox[0.8\linewidth]{
\State$(x_1, x_2,\ldots, x_n)=(a_1,a_2,\ldots,a_n)$
\While{$g_1 = \cdots = g_k =0\textbf{ and } h_1>0,\dotsc,h_s>0$}
\State $\begin{pmatrix}
x_1 \\
x_2 \\
\vdots \\
x_n
\end{pmatrix}
\xleftarrow{\textbf{F}}
\begin{pmatrix}
f_1\\
f_2\\
\vdots\\
f_n
\end{pmatrix}
$
\EndWhile
}

\noindent is called a \emph{semi-algebraic loop} on $S$ w.r.t. $F$. 
We denote by $\mathcal{S}(\gb, \mathbf{h})$ the solution set in $\mathbb{R}^n$ of the polynomial system defined by $\gb$ and $\mathbf{h}$.
\end{definition}

\noindent The following proposition is a direct consequence of the definitions.
\begin{proposition}\label{prop:saloop}
Let $\ab \in \R^n$, and $\gb$ and $\mathbf{h}$ be as above. %in Definition~\ref{def:saloop}.
Let $r_1,\dotsc,r_p$ be polynomial invariants of $\mathcal{L}(\ab,0,F)$, then
the semi-algebraic loop $\mathcal{L}(\ab,(\gb,\mathbf{h}),F)$ never terminates if
$
    \V(r_1,\dotsc,r_p) \cap \R^n \subset \mathcal{S}(\gb,\mathbf{h}).
$
\end{proposition}
The above inclusion corresponds to the quantified formula:
\[
\forall \xb \in \R^n, r_1(\xb)=\cdots=r_p(\xb)=0 \Rightarrow\left\{\hspace*{-0.2cm}\begin{array}{l}
g_1(\xb) = \cdots = g_k(\xb) =0\\h_1(\xb)>0,\dotsc,h_s(\xb)>0\end{array}\right..
\]
The validity of such a formula can be decided using a quantifier elimination algorithm \cite[Chapter 14]{bpr2006}. Since there are no free variables or alternating quantifiers, it corresponds to the emptiness decision of the set of solutions of a polynomial system of equations and inequalities. This is efficiently tackled by specific algorithms, whose most general version can be found in \cite[Theorem 13.24]{bpr2006}.
% Since there is no quantifier alternate, the complexity of such algorithms in this case is asymptotically optimal.
Besides, given the particular structure of this formula, an efficient approach would be to follow the one of \cite{synthesis2023algebro}, which is based on a combination of the Real Nullstellensatz \cite{BCR1998} and Putinar's Positivstellensatz \cite{Pu1993}.

We do not go further on these aspects as it diverges from the paper's focus and these directions will be explored in future works. Instead, we show 
%in %the following example, 
why the above sufficient criterion is not necessary. % one.
\begin{example}\label{exsaloop}
%\todoremi[inline]{Add one more variable}
Consider the elementary semi-algebraic loop from Appendix~\ref{bench}. 
%\programbox[0.4\linewidth]{
%\State$(x_1, x_2)=(a_1,a_2)$
%\While{$x_1>0$}
%\State $\begin{pmatrix}
%x_1 \\
%x_2
%\end{pmatrix}
%\longleftarrow
%\begin{pmatrix}
%2x_1\\
%2x_2
%\end{pmatrix}
%$
%\EndWhile
%}
A direct study of the linear recursive sequence defined by the successive values $\ab^0,\ab^1,\dotsc$ of $(x_1,x_2)$ shows that this loop never terminates if, and only if $a_1 >0$. Besides, $a_2x_1-a_1x_2=0$ is a polynomial invariant of this loop, and since every $\ab^j$, for $j\geq 0$ must be on this line, it generates the whole invariant ideal.
However, $\V(a_2x_1-a_1x_2)\cap \mathbb{R}^2$ is not contained in $\mathcal{S}(0,x_1)$ as shown in Figure~\ref{fig:saloop}.
\end{example}
% Besides, the invariant ideal for this loop is the ideal generated by $2x-2$.
% Hence, even if clearly $\V(2x-2) \cap $
% \begin{center}
\begin{comment}
    {\centering
\begin{figure}[H]
\includegraphics[width=0.55\linewidth]{saloop.pdf}
\caption{\textmd{An illustration of a particular case of Example~\ref{exsaloop} where $(a_1,a_2)=(2,1)$. In blue are depicted the successive values $\ab^0,\ab^1,\dotsc$ of the variables $(x_1,x_2)$, in red is the real zero-set of the invariant ideal, and in gray the set $\mathcal{S}(0,x_1)$ defined by the condition $x_1>0$.}}
\end{figure}
}
\end{comment}
% \end{center}

\vspace{-3mm}
%%%%%%%%%%%%%%%%%%%%%%%%%%%%%%%%%%%%%%%%%%%%%%%%%%%%%%%%%%
%%%%%%%%%%%%%%%%%%%%%%%%%%%%%%%%%%%%%%%%%%%%%%%%%%%%%%%%%%
\section{Implementation and Experiments
%Implementation and experimental results
}\label{sec:implementation}
%%%%%%%%%%%%%%%%%%%%%%%%%%%%%%%%%%%%%%%%%%%%%%%%%%%%%%%%%%
In this section, we present an implementation of the algorithms presented in this paper and compare its performances %on some examples from the literature. We 
with \textsf{Polar}~\cite{DBLP:journals/pacmpl/MoosbruggerSBK22}, which is mainly based on \cite{Unsolvableloops} for the case of unsolvable loops.
%\textcolor{red}{Example~\ref{ex2} is sourced from \cite{hrushovski2018polynomial}, ``Floor function'' is cited from \cite{rodriguez2004automatic}, while the remaining examples are from~\cite{amrollahi2023solvable}.}
%%%%%%%%%%%%%%%%%%%%%%%%%%%%%%%%%%%%%%%%%%%%%%%%%%%%%%%%%%
\subsection{Implementation details}
%We implemented Algorithms~\ref{algo1} and \ref{alg2} in Macaulay2~\cite{M2} and 
A prototype implementation of our algorithm in Macaulay2~\cite{M2} for polynomial loops with fixed initial values is publicly available.\footnote{\url{https://github.com/FatemehMohammadi/Algebraic_PolyLoop_Invariants.git}} The experiments are performed on a laptop equipped with a 4.8 GHz Intel i7 processor, 16 GB of RAM and 25 MB L3 cache. This prototype relies mainly on classic linear algebra routines and Gr\"obner bases computations available in Macaulay2. This makes the implementation quite direct from the pseudo-code presented here.
% All experiments were carried out on  The ``\VectorBasis'' procedure computes a vector basis for the solution set of linear equations using the kernel of the corresponding matrix, which can be obtained by the ``ker'' command in Macaulay~\cite{M2}.
We made slight modifications to these algorithms to speed up computations, based on observations from experiments, which we explain below. 

%Besides, we added some slight variations to these algorithms to speedup computations, according to observations from several experiments. We explain these below.

We first observed that most of the time, for simple loops, all the candidate polynomials in $\mathcal{B}$, computed at step~\ref{step:B} of Algorithm~\ref{alg2}, are actually polynomial invariants. Another observation, is that the smaller the dimension of the variety $X$ is, the faster Algorithm~\ref{algo1} computes polynomials defining $S_{(F,X)}$, for some polynomial map $F$. Hence, before checking them individually, we first test all elements of $\mathcal{B}$ at once. The collection of these polynomials defines a variety of smaller dimension than each individually, resulting in a potential speedup of up to 100 times (e.g., in Example~\ref{ex3}). %  ``Squares''). 

Another observation is that the value $\binom{n+d}{d}$ for $M$ is a rough overestimate corresponding to the worst case scenario (when no polynomial invariant exists). Users can adjust this parameter according to the specific example. Notably, in most cases, if one linear equation $g(F^k(\ab),\yb)$ is a linear combination of the others, the same applies to all subsequent equations $g(F^{l}(\ab),\yb)$ for $k<l\leq \binom{n+d}{d}$.

\vspace{-2mm}
\subsection{Experimental results}

% \todoremi[inline]{* Detail the computer algebra system used, the machinery that is used for elaborated subroutines (Gr\"obner basis computations, linear system solving, etc.).\\
% * Explain every heuristic (trick) that is used.\\
% * Mention that a further extension would be to use more efficient ad-hoc softwares (\href{https://msolve.lip6.fr}{\texttt{msolve}} \cite{BES2021})}\todoremi[inline]{* Explain the setting: the method you compare to, the example you consider, the data you show\\
% * the specifications of your laptop\\
% * Two tables: quantitative data, comparison with the existing literature.}
% Table~\ref{table 1} shows the dimensions of $I_{d,\mathcal{L}}$, indicating the computation of all polynomial invariants up to a specified degree. If a $*$ precedes a number, it indicates that not all invariants were computed within the 240-second timeframe of Algorithm~\ref{alg2}. After 240 seconds, we collected all polynomial invariants that were checked by Algorithm~\ref{algo1}. The first example is from \cite{rodriguez2004automatic}, while the others are from~\cite{Unsolvableloops}. 

%"0" denotes that Algorithm~\ref{algo1} was not used in Algorithm~\ref{alg2}.

%\review
{
In Tables~\ref{table 1} and \ref{table2}, we compare our implementation of \emph{Algorithm~\ref{alg2}} with the software \textsf{Polar}, %\footnote{\url{https://github.com/probing-lab/polar/}}, 
which is based on \cite{Unsolvableloops}, for the case of unsolvable loops. The benchmarks include those presented in \cite{Unsolvableloops} (with fixed initial values), as well as unsolvable loops in the last two rows, where \textsf{Polar} fails to find any polynomial invariant of degree less than 4.
}

%\review
{
Note that, unlike Algorithm~\ref{alg2}, \textsf{Polar} can handle loops with arbitrary initial values. However, on these benchmarks, it only produces invariants of the form $f(\xb)-f(\ab)=0$, where $\ab$ is the initial value, and $f\in\mathbb{Q}[\xb]$. Thanks to Proposition~\ref{prop4}, this can be achieved by our approach for a negligible additional cost, allowing us to determine whether such an invariant exists or not.
%by first fixing some initial value before applying Algorithm~\ref{alg2}, and then Algorithm~\ref{algo1}.
}
%\review
{
A notable distinction is that our approach is global as we compute \emph{all possible polynomial invariants up to a specified degree}, whereas \textsf{Polar} generates only a (possibly empty) \emph{subset of them}. However, \textsf{Polar} can handle probabilistic loops, whereas ours is limited to deterministic ones.
}

%\review
{
In Table~\ref{table 1}, we report quantitative data on the output of Algorithm~\ref{alg2} for various benchmarks (listed in the rows) and for generating polynomial invariants of degree $1,2,3,$ and $4$. Additionally, the column $d$ denotes the number of polynomials outputted by Algorithm~\ref{alg2}, which is also the dimension of the associated truncated invariant ideal. Finally, in the column \textsf{Polar}, we report the number of polynomial invariants computed by \textsf{Polar}.
}
%\review
{
In Table~\ref{table2}, we present the timings corresponding to the executions of Algorithm~\ref{alg2} as reported in Table~\ref{table 1}. The timings are in seconds, and we chose a time limit of 360 seconds. In cases where \textsf{Polar} reaches this time limit (degree 4 for Yaghzhev9), we ensured that it does not terminate after 15 minutes or reach maximum recursion depth.
}

% \subsection{Discussion}

%Algorithm~\ref{alg2} excels in cases where the dimension of truncated ideals is small.
%In certain examples, we use the flexibility of our algorithm to find polynomial invariants : we compute all polynomial invariants with fixed support (the set of non-zero coefficients), instead of the one up to a some degree. 
%In most examples, the validity of these candidates as polynomial invariants is verified. Notably, Gr\"obner bases computations
%which exhibits doubly exponential complexity, 
%is required for each iteration in Algorithm~\ref{algo1} to address the ideal membership problem. However, Table~\ref{table 1} indicates that the number of iterations denoted by $N$ within Algorithm~\ref{algo1} is quite small. 
%\review{This is mainly thanks to the new termination condition, for radical ideals, in the while loop of Algorithm~\ref{algo1}.}

%%%%%%%%%%%%%%%%%%%%%%%%%%%%%%%%%%%%%
%Squares:
%1+x+y+z=0
% yagzhev9:
%$x_1-x_3+x_5=x_2-x_4+x_6=-x_7+x_8-7=0$
%Example 10:
%$(3a_1-a_2-4a_3)(x_1+x_2)-(3a_1-a_2-4a_3)(x_2+x_3)-(a_1-a_3)(x_1+x_2)^2-16(a_1-a_3)(x_2+x_3)^2-24(a_1-a_3)(x_1+x_2)(x_2+x_3).$
%%%%%%%%%%%%%%%%%%%%%%%%%%%%%%%%%%%%%%%%%%%%%%%%%
%\vspace{-2mm}

%\review
%\vspace{-6mm}{
We remark that in Table~\ref{algo1}, when there exists non-zero polynomial invariants, we almost always find more than \textsf{Polar}. For example, in the linear case, for Squares, we find the single invariant $1+x_1+x_2+x_3=0$ (see also Example~\ref{ex3}). Additionally, for Yaghzev9, we find $x_1-x_3+x_5=0$, $x_2-x_4+x_6=0$, and $x_8-x_7-7=0$.
Note also that for the case of Ex\,10, \textsf{Polar} fails to find the following ``general'' invariant in terms of the initial values $(a_1,a_2,a_3)$: 
%\vspace{-3mm}
\begin{flushleft}
$(3a_1-a_2-4a_3)^2(x_1+x_2)-(3a_1-a_2-4a_3)^2(x_2+x_3)-9(a_1-a_3)(x_1+x_2)^2-16(a_1-a_3)(x_2+x_3)^2+   24(a_1-a_3)(x_1+x_2)(x_2+x_3)=0$.
\end{flushleft}
%}
%\vspace{-4mm}

%\review

However, for the Yagzhev9 and Yagzhev11 examples (with 9 and 11 variables, respectively), \textsf{Polar} handles degree 3, unlike our approach. Yet, both reach the fixed time limit at degree 4. Comparing timings in Table~\ref{table2}, our approach outperforms \textsf{Polar} for small degrees, but \textsf{Polar} demonstrates better performance for larger degrees and variables. However, since our method is complete, this increase in complexity is unavoidable.
%\\
%However, for the Yagzhev9 and Yagzhev11 benchmarks (with 9 and 11 variables, respectively), \textsf{Polar}, unlike our approach, is capable of handling degree 3. However, the fixed time limit for both is reached at degree 4. Indeed, comparing timings in Table~\ref{table2}, our approach outperforms \textsf{Polar} for small degrees, but the latter demonstrates better performance for larger degrees and variables. However, since our method is complete, this increase in complexity is unavoidable.}
%the approach in~\cite{Unsolvableloops} cannot guarantee that it generates all polynomial invariants up to a given degree, while Algorithm~\ref{alg2} ensures that the computed polynomial invariants form a basis for truncated ideals. 
%\review

In particular, even when the output is empty (i.e. $d=0$ in Table~\ref{table2}), or the same as \textsf{Polar} (the value of $d$ is the same in the column \textsf{Polar}), we can conclude that no  additional (and linearly independent) polynomial invariant can be found.

%Algorithm~\ref{alg2} works very well when the dimension of truncated ideals are small. For some examples, we use different implementation for computing polynomial invariants. For computing $I_{3,\mathcal{L}}$ of Example~\ref{ex12}, we use specific form for choosing polynomial invariant of degree 3 which implies that we don't compute all polynomial invariants up to degree 3 but we compute all polynomial invariants of the given form. For the most of the examples, it is verified that all candidates are actually polynomial invariants. A Gr\"obner basis computation, which is doubly exponential, is needed for every iteration in Algorithm~\ref{algo1} to solve ideal membership problem but from Table~\ref{table 1}, we observe that number of iterations "$N_{max}$" used Algorithm~\ref{algo1} inside Algorithm~\ref{alg2} is quiet small . From Table~\ref{table2}, timings of our approach for small degrees are much smaller that timings of~\cite{Unsolvableloops} for small degrees but for large degrees, timings of~\cite{Unsolvableloops} is better than us. However, the approach in\cite{Unsolvableloops} cannot verify that they generate all polynomial invariants up to a given degree whilst we verify polynomial invariants computed by Algorithm~\ref{alg2} form a basis for truncated ideals. Moreover,  even we give same polynomial invariants as~\cite{Unsolvableloops}, we verify that there is no polynomial invariant outside of the vector space generated by computed polynomial invariants.

\begin{table}[H]
    \centering
    \scalebox{0.9}{
\begin{tabular}{|*{9}{c|}}
    \hline
        \multicolumn{1}{|c}{Degree} & \multicolumn{2}{|c}{1} & \multicolumn{2}{|c}{2} &  \multicolumn{2}{|c}{3} & \multicolumn{2}{|c|}{4} \\ \hline
        Benchmark & d & \textsf{Polar}&  d & \textsf{Polar}& d& \textsf{Polar}&  d &\textsf{Polar} \\ \hline
       
         Fib1  & 0 &0 & 0&0  & 1 &1& \textbf{ 4}&1 \\ \hline
        Fib2  & 0& 0&0 & 0  &  1 &1& {TL}&\textbf{1} \\ \hline
        Fib3   & 0 &0 & 0 &0&  1 &1&  \textbf{4}&1 \\ \hline
        Nagata  & \textbf{1}&0 &  \textbf{5}&1  & \textbf{13}&1 &  \textbf{26}&2 \\ \hline
        Yagzhev9  & \textbf{3}&0  & TL &\textbf{3} & {TL}&\textbf{3} &  {TL}&TL \\ \hline
        Yagzhev11  & 0 &0  & 0 &0  & TL&\textbf{1}  & {TL}&TL \\ \hline
        Ex 9 & 0 &0  & 0 &0 & \textbf{3}&1 &  \textbf{11}&1 \\ \hline 
        Ex 10  & 0 &0 & \textbf{2} &0 & \textbf{8} &0& \textbf{19}&0
        \\ \hline
         Squares (Ex~\ref{ex3}) & \textbf{1} & 0 & \textbf{5} &0 &  \textbf{13  }& 0& \textbf{26} &0\\ \hline
        
\end{tabular}
}

{\small TL = Timeout (360 seconds); \quad
\textbf{bold}: new invariants found}
\caption{\label{table 1}Data on outputs of Algorithm~\ref{alg2} and \textsf{Polar}
%d= the dimension of the corresponding truncated ideal; $N$= the largest number of iterations used in Algorithm~\ref{algo1} inside Algorithm~\ref{alg2}; 
%\review{bold numbers indicates maximal number of the number of polynomial invariants computed by each approach }.
}
\end{table}

\vspace*{-3em}

\begin{table}[H]
    \centering\hspace*{-0.2cm}
    \scalebox{0.75}{
    \begin{tabular}{|*{9}{c|}}
    \hline
        \multicolumn{1}{|c}{Degree} & \multicolumn{2}{|c}{1} & \multicolumn{2}{|c}{2} &  \multicolumn{2}{|c}{3} & \multicolumn{2}{|c|}{4} \\ \hline
        Benchmark & Ours &\textsf{Polar} & Ours & \textsf{Polar} & Ours & \textsf{Polar} & Ours & \textsf{Polar} \\ \hline      
        Fib1 & \textbf{0.014} & 0.2 & \textbf{0.046} & 0.32 & \textbf{0.17} & 0.68 & 37.07 & \textbf{1.58} \\ \hline
        Fib2 & \textbf{0.017} & 0.23 & \textbf{0.056} & 0.46 & 12.62 & \textbf{1.18} & {TL} & \textbf{3.69} \\ \hline
        Fib3 & \textbf{0.013} & 0.21 & \textbf{0.056} & 0.4 & \textbf{0.225} & 1.18 & 7.11 & \textbf{3.82} \\ \hline
        Nagata& \textbf{0.014} & 0.25 & \textbf{0.057} & 0.55 & \textbf{0.09} & 1.21 & \textbf{0.19} & 2.84 \\ \hline
        Yagzhev9 & 1.46 & \textbf{0.43} & TL & \textbf{5.2}& {TL} & \textbf{131.5} & {TL} & {TL} \\ \hline
        Yagzhev11 & \textbf{0.075} & 0.45 & 61.4 & \textbf{6.83} &TL & \textbf{359} & {TL}& {TL}\\ \hline
        Ex 9& \textbf{0.016} & 0.28 & \textbf{0.06} & 0.64 & \textbf{0.19} & 2.38 & \textbf{0.55} & 11.5 \\ \hline
        Ex 10 & \textbf{0.014} & 0.51 & \textbf{0.058} & 0.7 & \textbf{0.16} & 1.21 & \textbf{0.45} & 2.3 \\ \hline
          \parbox{0.8cm}{\centering Squares\\[-0.4em](Ex~\ref{ex3})}
        & \textbf{0.0151}  & 0.5& \textbf{0.085} & 0.67 & \textbf{0.25}
        & 1.15 & \textbf{1.6} & 2.25 \\ \hline
       
    \end{tabular}
    }
    
    {\small TL = Timeout   (360 seconds); }
    \caption{\label{table2}Timings for Algorithm~\ref{alg2} and \textsf{Polar}, in seconds
    %\normalfont The computation times in seconds required to obtain a basis for truncated invariant ideals at their respective degrees of a loop \review{ with given initial values; bold numbers indicates the best timings of two approaches}.
    }
\end{table}

 {\centering
\begin{figure}[H]
\includegraphics[width=0.55\linewidth]{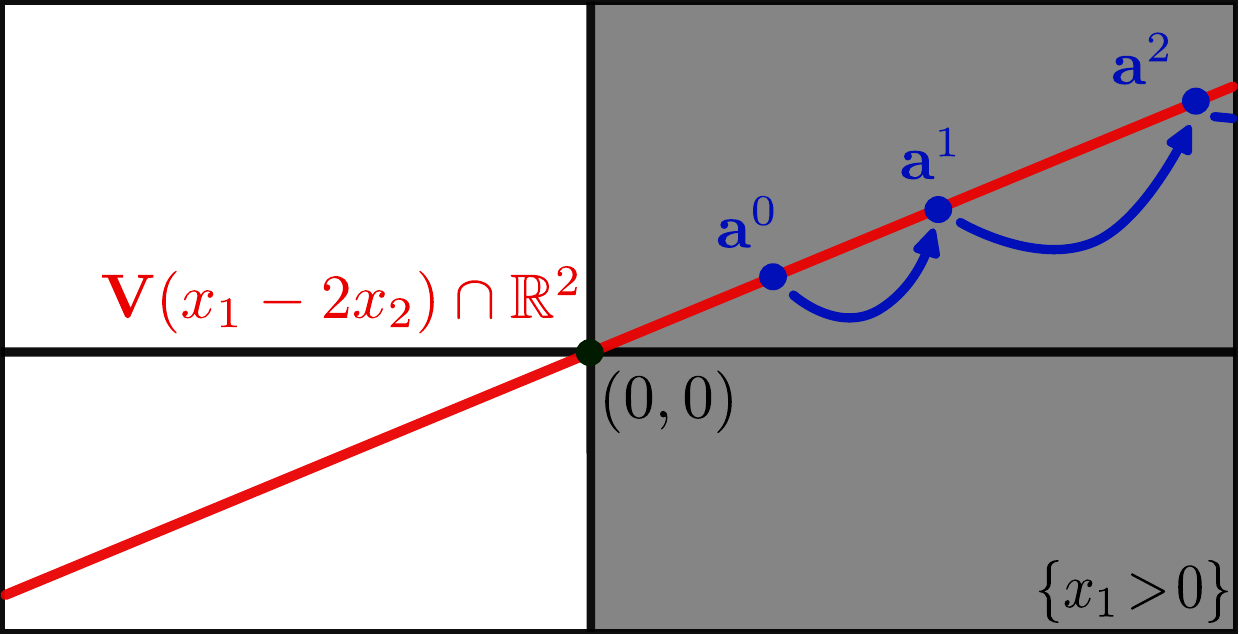}
\caption{\textmd{An illustration of a particular case of Example~\ref{exsaloop} where $(a_1,a_2)=(2,1)$. In blue are depicted the successive values $\ab^0,\ab^1,\dotsc$ of the variables $(x_1,x_2)$, in red is the real zero-set of the invariant ideal, and in gray the set $\mathcal{S}(0,x_1)$ defined by the condition $x_1>0$.}}\label{fig:saloop}
\end{figure}
}

\color{black}
\bibliographystyle{ACM-Reference-Format}
\bibliography{ISAAC_References}


\begin{thebibliography}{36}

%%% the Web address.

\ifx \showCODEN    \undefined \def \showCODEN     #1{\unskip}     \fi
\ifx \showDOI      \undefined \def \showDOI       #1{#1}\fi
\ifx \showISBNx    \undefined \def \showISBNx     #1{\unskip}     \fi
\ifx \showISBNxiii \undefined \def \showISBNxiii  #1{\unskip}     \fi
\ifx \showISSN     \undefined \def \showISSN      #1{\unskip}     \fi
\ifx \showLCCN     \undefined \def \showLCCN      #1{\unskip}     \fi
\ifx \shownote     \undefined \def \shownote      #1{#1}          \fi
\ifx \showarticletitle \undefined \def \showarticletitle #1{#1}   \fi
\ifx \showURL      \undefined \def \showURL       {\relax}        \fi
% invisible to TeX
\providecommand\bibfield[2]{#2}
\providecommand\bibinfo[2]{#2}
\providecommand\natexlab[1]{#1}
\providecommand\showeprint[2][]{arXiv:#2}

\bibitem[Amrollahi et~al\mbox{.}(2022)]%
        {Unsolvableloops}
\bibfield{author}{\bibinfo{person}{Daneshvar Amrollahi}, \bibinfo{person}{Ezio
  Bartocci}, \bibinfo{person}{George Kenison}, \bibinfo{person}{Laura
  Kov{\'a}cs}, \bibinfo{person}{Marcel Moosbrugger}, {and}
  \bibinfo{person}{Miroslav Stankovi{\v{c}}}.} \bibinfo{year}{2022}\natexlab{}.
\newblock \showarticletitle{Solving invariant generation for unsolvable loops}.
  In \bibinfo{booktitle}{\emph{International Static Analysis Symposium}}.
  \bibinfo{publisher}{Springer}, \bibinfo{pages}{19--43}.
\newblock


\bibitem[Amrollahi et~al\mbox{.}(2023)]%
        {amrollahi2023solvable}
\bibfield{author}{\bibinfo{person}{Daneshvar Amrollahi}, \bibinfo{person}{Ezio
  Bartocci}, \bibinfo{person}{George Kenison}, \bibinfo{person}{Laura
  Kov{\'a}cs}, \bibinfo{person}{Marcel Moosbrugger}, {and}
  \bibinfo{person}{Miroslav Stankovi{\v{c}}}.} \bibinfo{year}{2023}\natexlab{}.
\newblock \showarticletitle{(Un) Solvable Loop Analysis}.
\newblock \bibinfo{journal}{\emph{arXiv preprint arXiv:2306.01597}}
  (\bibinfo{year}{2023}).
\newblock


\bibitem[Baake et~al\mbox{.}(1993)]%
        {baake1993trace}
\bibfield{author}{\bibinfo{person}{Michael Baake}, \bibinfo{person}{Uwe Grimm},
  {and} \bibinfo{person}{Dieter Joseph}.} \bibinfo{year}{1993}\natexlab{}.
\newblock \showarticletitle{Trace maps, invariants, and some of their
  applications}.
\newblock \bibinfo{journal}{\emph{International Journal of Modern Physics B}}
  \bibinfo{volume}{7}, \bibinfo{number}{06n07} (\bibinfo{year}{1993}),
  \bibinfo{pages}{1527--1550}.
\newblock


\bibitem[Basu et~al\mbox{.}(2006)]%
        {bpr2006}
\bibfield{author}{\bibinfo{person}{Saugata Basu}, \bibinfo{person}{Richard
  Pollack}, {and} \bibinfo{person}{Marie-Françoise Roy}.}
  \bibinfo{year}{2006}\natexlab{}.
\newblock \bibinfo{booktitle}{\emph{Algorithms in Real Algebraic Geometry}
  (\bibinfo{edition}{2nd revised and extended 2016} ed.)}.
\newblock \bibinfo{publisher}{Springer International Publishing}.
\newblock
\urldef\tempurl%
\url{https://doi.org/10.1007/3-540-33099-2}
\showDOI{\tempurl}


\bibitem[Berthomieu et~al\mbox{.}(2021)]%
        {berthomieu2021msolve}
\bibfield{author}{\bibinfo{person}{J{\'e}r{\'e}my Berthomieu},
  \bibinfo{person}{Christian Eder}, {and} \bibinfo{person}{Mohab {Safey El
  Din}}.} \bibinfo{year}{2021}\natexlab{}.
\newblock \showarticletitle{Msolve: A library for solving polynomial systems}.
  In \bibinfo{booktitle}{\emph{Proceedings of the 2021 on International
  Symposium on Symbolic and Algebraic Computation}}. \bibinfo{pages}{51--58}.
\newblock


\bibitem[Bochnack et~al\mbox{.}(1998)]%
        {BCR1998}
\bibfield{author}{\bibinfo{person}{Jan Bochnack}, \bibinfo{person}{Michel
  Coste}, {and} \bibinfo{person}{Marie-Françoise Roy}.}
  \bibinfo{year}{1998}\natexlab{}.
\newblock \bibinfo{booktitle}{\emph{Real Algebraic Geometry}
  (\bibinfo{edition}{1st} ed.)}. \bibinfo{series}{Ergebnisse der Mathematik und
  ihrer Grenzgebiete}, Vol.~\bibinfo{volume}{3}.
\newblock \bibinfo{publisher}{Springer-Verlag}, \bibinfo{address}{Berlin,
  Heidelberg}.
\newblock
\urldef\tempurl%
\url{https://doi.org/10.1007/978-3-642-85463-2}
\showDOI{\tempurl}


\bibitem[Chatterjee et~al\mbox{.}(2020)]%
        {chatterjee2020polynomial}
\bibfield{author}{\bibinfo{person}{Krishnendu Chatterjee},
  \bibinfo{person}{Hongfei Fu}, \bibinfo{person}{Amir~Kafshdar Goharshady},
  {and} \bibinfo{person}{Ehsan~Kafshdar Goharshady}.}
  \bibinfo{year}{2020}\natexlab{}.
\newblock \showarticletitle{Polynomial invariant generation for
  non-deterministic recursive programs}. In
  \bibinfo{booktitle}{\emph{Proceedings of the 41st ACM SIGPLAN Conference on
  Programming Language Design and Implementation}}. \bibinfo{pages}{672--687}.
\newblock


\bibitem[Cox et~al\mbox{.}(2013)]%
        {cox2013ideals}
\bibfield{author}{\bibinfo{person}{David Cox}, \bibinfo{person}{John Little},
  {and} \bibinfo{person}{Donal O'Shea}.} \bibinfo{year}{2013}\natexlab{}.
\newblock \bibinfo{booktitle}{\emph{Ideals, varieties, and algorithms: an
  introduction to computational algebraic geometry and commutative algebra}}.
\newblock \bibinfo{publisher}{Springer Science \& Business Media}.
\newblock


\bibitem[Cyphert and Kincaid(2024)]%
        {cyphert2023solvable}
\bibfield{author}{\bibinfo{person}{John Cyphert} {and} \bibinfo{person}{Zachary
  Kincaid}.} \bibinfo{year}{2024}\natexlab{}.
\newblock \showarticletitle{Solvable Polynomial Ideals: The Ideal Reflection
  for Program Analysis}.
\newblock \bibinfo{journal}{\emph{Proceedings of the ACM on Programming
  Languages}} \bibinfo{volume}{8}, \bibinfo{number}{POPL}
  (\bibinfo{year}{2024}), \bibinfo{pages}{724--752}.
\newblock


\bibitem[de~Oliveira et~al\mbox{.}(2016)]%
        {de2016polynomial}
\bibfield{author}{\bibinfo{person}{Steven de Oliveira}, \bibinfo{person}{Saddek
  Bensalem}, {and} \bibinfo{person}{Virgile Prevosto}.}
  \bibinfo{year}{2016}\natexlab{}.
\newblock \showarticletitle{Polynomial invariants by linear algebra}. In
  \bibinfo{booktitle}{\emph{Automated Technology for Verification and Analysis:
  14th International Symposium, ATVA 2016, Chiba, Japan, October 17-20, 2016,
  Proceedings 14}}. Springer, \bibinfo{pages}{479--494}.
\newblock


\bibitem[de~Oliveira et~al\mbox{.}(2017)]%
        {de2017synthesizing}
\bibfield{author}{\bibinfo{person}{Steven de Oliveira}, \bibinfo{person}{Saddek
  Bensalem}, {and} \bibinfo{person}{Virgile Prevosto}.}
  \bibinfo{year}{2017}\natexlab{}.
\newblock \showarticletitle{Synthesizing invariants by solving solvable loops}.
  In \bibinfo{booktitle}{\emph{International Symposium on Automated Technology
  for Verification and Analysis}}. Springer, \bibinfo{pages}{327--343}.
\newblock


\bibitem[Eder and Faug{\`e}re(2017)]%
        {eder2017survey}
\bibfield{author}{\bibinfo{person}{Christian Eder} {and}
  \bibinfo{person}{Jean-Charles Faug{\`e}re}.} \bibinfo{year}{2017}\natexlab{}.
\newblock \showarticletitle{A survey on signature-based algorithms for
  computing Gr{\"o}bner bases}.
\newblock \bibinfo{journal}{\emph{Journal of Symbolic Computation}}
  \bibinfo{volume}{80} (\bibinfo{year}{2017}), \bibinfo{pages}{719--784}.
\newblock


\bibitem[Floyd(1993)]%
        {floyd1993assigning}
\bibfield{author}{\bibinfo{person}{Robert~W Floyd}.}
  \bibinfo{year}{1993}\natexlab{}.
\newblock \showarticletitle{Assigning meanings to programs}.
\newblock In \bibinfo{booktitle}{\emph{Program Verification: Fundamental Issues
  in Computer Science}}. \bibinfo{publisher}{Springer},
  \bibinfo{pages}{65--81}.
\newblock


\bibitem[Goharshady et~al\mbox{.}(2023)]%
        {synthesis2023algebro}
\bibfield{author}{\bibinfo{person}{Amir~Kafshdar Goharshady},
  \bibinfo{person}{S Hitarth}, \bibinfo{person}{Fatemeh Mohammadi}, {and}
  \bibinfo{person}{Harshit~Jitendra Motwani}.} \bibinfo{year}{2023}\natexlab{}.
\newblock \showarticletitle{Algebro-geometric Algorithms for Template-based
  Synthesis of Polynomial Programs}.
\newblock \bibinfo{journal}{\emph{Proceedings of the ACM on Programming
  Languages}} \bibinfo{volume}{7}, \bibinfo{number}{OOPSLA1}
  (\bibinfo{year}{2023}), \bibinfo{pages}{727--756}.
\newblock


\bibitem[Grayson and Stillman({[n.\,d.]})]%
        {M2}
\bibfield{author}{\bibinfo{person}{Daniel~R. Grayson} {and}
  \bibinfo{person}{Michael~E. Stillman}.} \bibinfo{year}{[n.\,d.]}\natexlab{}.
\newblock \bibinfo{title}{Macaulay2, a software system for research in
  algebraic geometry}.
\newblock \bibinfo{howpublished}{Available at \url{http://www2.macaulay2.com}}.
\newblock


\bibitem[Hrushovski et~al\mbox{.}(2018)]%
        {hrushovski2018polynomial}
\bibfield{author}{\bibinfo{person}{Ehud Hrushovski}, \bibinfo{person}{Jo{\"e}l
  Ouaknine}, \bibinfo{person}{Amaury Pouly}, {and} \bibinfo{person}{James
  Worrell}.} \bibinfo{year}{2018}\natexlab{}.
\newblock \showarticletitle{Polynomial invariants for affine programs}. In
  \bibinfo{booktitle}{\emph{Proceedings of the 33rd Annual ACM/IEEE Symposium
  on Logic in Computer Science}}. \bibinfo{pages}{530--539}.
\newblock


\bibitem[Hrushovski et~al\mbox{.}(2023)]%
        {hrushovski2023strongest}
\bibfield{author}{\bibinfo{person}{Ehud Hrushovski}, \bibinfo{person}{Jo{\"e}l
  Ouaknine}, \bibinfo{person}{Amaury Pouly}, {and} \bibinfo{person}{James
  Worrell}.} \bibinfo{year}{2023}\natexlab{}.
\newblock \showarticletitle{On strongest algebraic program invariants}.
\newblock \bibinfo{journal}{\emph{J. ACM}} \bibinfo{volume}{70},
  \bibinfo{number}{5} (\bibinfo{year}{2023}), \bibinfo{pages}{1--22}.
\newblock


\bibitem[Karr(1976)]%
        {karr1976affine}
\bibfield{author}{\bibinfo{person}{Michael Karr}.}
  \bibinfo{year}{1976}\natexlab{}.
\newblock \showarticletitle{Affine relationships among variables of a program}.
\newblock \bibinfo{journal}{\emph{Acta Informatica}}  \bibinfo{volume}{6}
  (\bibinfo{year}{1976}), \bibinfo{pages}{133--151}.
\newblock


\bibitem[Kauers and Zimmermann(2008)]%
        {kauers2008computing}
\bibfield{author}{\bibinfo{person}{Manuel Kauers} {and}
  \bibinfo{person}{Burkhard Zimmermann}.} \bibinfo{year}{2008}\natexlab{}.
\newblock \showarticletitle{Computing the algebraic relations of C-finite
  sequences and multisequences}.
\newblock \bibinfo{journal}{\emph{Journal of Symbolic Computation}}
  \bibinfo{volume}{43}, \bibinfo{number}{11} (\bibinfo{year}{2008}),
  \bibinfo{pages}{787--803}.
\newblock


\bibitem[Kempf(1993)]%
        {kempf_1993}
\bibfield{author}{\bibinfo{person}{G. Kempf}.} \bibinfo{year}{1993}\natexlab{}.
\newblock \bibinfo{booktitle}{\emph{Algebraic Varieties}}.
\newblock \bibinfo{publisher}{Cambridge University Press}.
\newblock
\urldef\tempurl%
\url{https://doi.org/10.1017/CBO9781107359956}
\showDOI{\tempurl}


\bibitem[Kov{\'a}cs(2008)]%
        {kovacs2008reasoning}
\bibfield{author}{\bibinfo{person}{Laura Kov{\'a}cs}.}
  \bibinfo{year}{2008}\natexlab{}.
\newblock \showarticletitle{Reasoning algebraically about p-solvable loops}. In
  \bibinfo{booktitle}{\emph{International Conference on Tools and Algorithms
  for the Construction and Analysis of Systems}}. Springer,
  \bibinfo{pages}{249--264}.
\newblock


\bibitem[Kov{\'a}cs(2023)]%
        {kovacs2023algebra}
\bibfield{author}{\bibinfo{person}{Laura Kov{\'a}cs}.}
  \bibinfo{year}{2023}\natexlab{}.
\newblock \showarticletitle{Algebra-Based Loop Analysis}. In
  \bibinfo{booktitle}{\emph{Proceedings of the 2023 International Symposium on
  Symbolic and Algebraic Computation}}. \bibinfo{pages}{41--42}.
\newblock


\bibitem[Manna and Pnueli(2012)]%
        {manna2012temporal}
\bibfield{author}{\bibinfo{person}{Zohar Manna} {and} \bibinfo{person}{Amir
  Pnueli}.} \bibinfo{year}{2012}\natexlab{}.
\newblock \bibinfo{booktitle}{\emph{Temporal verification of reactive systems:
  safety}}.
\newblock \bibinfo{publisher}{Springer Science \& Business Media}.
\newblock


\bibitem[Mayr and Meyer(1982)]%
        {mayr1982complexity}
\bibfield{author}{\bibinfo{person}{E. Mayr} {and} \bibinfo{person}{A. Meyer}.}
  \bibinfo{year}{1982}\natexlab{}.
\newblock \showarticletitle{The complexity of the word problem for commutative
  semi-groups and polynomial ideals}.
\newblock \bibinfo{journal}{\emph{Advances in Mathematics}}
  \bibinfo{volume}{46} (\bibinfo{year}{1982}), \bibinfo{pages}{305--329}.
\newblock
\urldef\tempurl%
\url{https://doi.org/10.1016/0001-8708(82)90035-9}
\showDOI{\tempurl}


\bibitem[Moosbrugger et~al\mbox{.}(2022)]%
        {DBLP:journals/pacmpl/MoosbruggerSBK22}
\bibfield{author}{\bibinfo{person}{Marcel Moosbrugger},
  \bibinfo{person}{Miroslav Stankovic}, \bibinfo{person}{Ezio Bartocci}, {and}
  \bibinfo{person}{Laura Kov{\'{a}}cs}.} \bibinfo{year}{2022}\natexlab{}.
\newblock \showarticletitle{This is the moment for probabilistic loops}.
\newblock \bibinfo{journal}{\emph{Proc. {ACM} Program. Lang.}}
  \bibinfo{volume}{6}, \bibinfo{number}{{OOPSLA2}} (\bibinfo{year}{2022}),
  \bibinfo{pages}{1497--1525}.
\newblock
\urldef\tempurl%
\url{https://doi.org/10.1145/3563341}
\showDOI{\tempurl}


\bibitem[Moreno~Socias(1992)]%
        {MorenoSocias_1992}
\bibfield{author}{\bibinfo{person}{Guillermo Moreno~Socias}.}
  \bibinfo{year}{1992}\natexlab{}.
\newblock \showarticletitle{Length of Polynomial Ascending Chains and Primitive
  Recursiveness.}
\newblock \bibinfo{journal}{\emph{MATHEMATICA SCANDINAVICA}}
  \bibinfo{volume}{71} (\bibinfo{date}{Jun.} \bibinfo{year}{1992}),
  \bibinfo{pages}{181–205}.
\newblock
\urldef\tempurl%
\url{https://doi.org/10.7146/math.scand.a-12421}
\showDOI{\tempurl}


\bibitem[M{\"u}ller-Olm and Seidl(2004a)]%
        {muller2004computing}
\bibfield{author}{\bibinfo{person}{Markus M{\"u}ller-Olm} {and}
  \bibinfo{person}{Helmut Seidl}.} \bibinfo{year}{2004}\natexlab{a}.
\newblock \showarticletitle{Computing polynomial program invariants}.
\newblock \bibinfo{journal}{\emph{Inform. Process. Lett.}}
  \bibinfo{volume}{91}, \bibinfo{number}{5} (\bibinfo{year}{2004}),
  \bibinfo{pages}{233--244}.
\newblock


\bibitem[M{\"u}ller-Olm and Seidl(2004b)]%
        {muller2004note}
\bibfield{author}{\bibinfo{person}{Markus M{\"u}ller-Olm} {and}
  \bibinfo{person}{Helmut Seidl}.} \bibinfo{year}{2004}\natexlab{b}.
\newblock \showarticletitle{A note on Karr’s algorithm}. In
  \bibinfo{booktitle}{\emph{International Colloquium on Automata, Languages,
  and Programming}}. Springer, \bibinfo{pages}{1016--1028}.
\newblock


\bibitem[Pastuszak(2020)]%
        {pastuszak2020ascending}
\bibfield{author}{\bibinfo{person}{Grzegorz Pastuszak}.}
  \bibinfo{year}{2020}\natexlab{}.
\newblock \showarticletitle{Ascending chains of ideals in the polynomial ring}.
\newblock \bibinfo{journal}{\emph{Turkish Journal of Mathematics}}
  \bibinfo{volume}{44}, \bibinfo{number}{6} (\bibinfo{year}{2020}),
  \bibinfo{pages}{2652--2658}.
\newblock
\urldef\tempurl%
\url{https://doi.org/10.3906/mat-1904-61}
\showDOI{\tempurl}


\bibitem[Putinar(1993)]%
        {Pu1993}
\bibfield{author}{\bibinfo{person}{Mihai Putinar}.}
  \bibinfo{year}{1993}\natexlab{}.
\newblock \showarticletitle{Positive Polynomials on Compact Semi-algebraic
  Sets}.
\newblock \bibinfo{journal}{\emph{Indiana University Mathematics Journal}}
  \bibinfo{volume}{42}, \bibinfo{number}{3} (\bibinfo{year}{1993}),
  \bibinfo{pages}{969--984}.
\newblock
\urldef\tempurl%
\url{https://doi.org/stable/24897130}
\showDOI{\tempurl}


\bibitem[Roberts and Baake(1994)]%
        {roberts1994trace}
\bibfield{author}{\bibinfo{person}{John~AG Roberts} {and}
  \bibinfo{person}{Michael Baake}.} \bibinfo{year}{1994}\natexlab{}.
\newblock \showarticletitle{Trace maps as 3D reversible dynamical systems with
  an invariant}.
\newblock \bibinfo{journal}{\emph{Journal of statistical physics}}
  \bibinfo{volume}{74}, \bibinfo{number}{3-4} (\bibinfo{year}{1994}),
  \bibinfo{pages}{829--888}.
\newblock


\bibitem[Rodr{\'\i}guez-Carbonell and Kapur(2004)]%
        {rodriguez2004automatic}
\bibfield{author}{\bibinfo{person}{Enric Rodr{\'\i}guez-Carbonell} {and}
  \bibinfo{person}{Deepak Kapur}.} \bibinfo{year}{2004}\natexlab{}.
\newblock \showarticletitle{Automatic generation of polynomial loop invariants:
  Algebraic foundations}. In \bibinfo{booktitle}{\emph{Proceedings of the 2004
  international symposium on Symbolic and algebraic computation}}.
  \bibinfo{pages}{266--273}.
\newblock


\bibitem[Rodr{\'\i}guez-Carbonell and Kapur(2007a)]%
        {rodriguez2007automatic}
\bibfield{author}{\bibinfo{person}{Enric Rodr{\'\i}guez-Carbonell} {and}
  \bibinfo{person}{Deepak Kapur}.} \bibinfo{year}{2007}\natexlab{a}.
\newblock \showarticletitle{Automatic generation of polynomial invariants of
  bounded degree using abstract interpretation}.
\newblock \bibinfo{journal}{\emph{Science of Computer Programming}}
  \bibinfo{volume}{64}, \bibinfo{number}{1} (\bibinfo{year}{2007}),
  \bibinfo{pages}{54--75}.
\newblock


\bibitem[Rodr{\'\i}guez-Carbonell and Kapur(2007b)]%
        {rodriguez2007generating}
\bibfield{author}{\bibinfo{person}{Enric Rodr{\'\i}guez-Carbonell} {and}
  \bibinfo{person}{Deepak Kapur}.} \bibinfo{year}{2007}\natexlab{b}.
\newblock \showarticletitle{Generating all polynomial invariants in simple
  loops}.
\newblock \bibinfo{journal}{\emph{Journal of Symbolic Computation}}
  \bibinfo{volume}{42}, \bibinfo{number}{4} (\bibinfo{year}{2007}),
  \bibinfo{pages}{443--476}.
\newblock


\bibitem[Shafarevich and Reid(1994)]%
        {shafarevich1994basic}
\bibfield{author}{\bibinfo{person}{Igor~Rostislavovich Shafarevich} {and}
  \bibinfo{person}{Miles Reid}.} \bibinfo{year}{1994}\natexlab{}.
\newblock \bibinfo{booktitle}{\emph{Basic algebraic geometry}}.
  Vol.~\bibinfo{volume}{1}.
\newblock \bibinfo{publisher}{Springer}.
\newblock


\bibitem[Von Zur~Gathen and Gerhard(2013)]%
        {von2013modern}
\bibfield{author}{\bibinfo{person}{Joachim Von Zur~Gathen} {and}
  \bibinfo{person}{J{\"u}rgen Gerhard}.} \bibinfo{year}{2013}\natexlab{}.
\newblock \bibinfo{booktitle}{\emph{Modern computer algebra}}.
\newblock \bibinfo{publisher}{Cambridge university press}.
\newblock


\end{thebibliography}

\appendix

\section{Examples and Benchmarks}\label{bench}

\vspace*{1em}
\begin{itemize}[label=]%\label{loops}
\item \boxed{Squares}\\[-0.14em]
\programboxappendix[0.55\linewidth]{
\State $(x_1,x_2, x_3)=(-1,-1, 1)$
\While{true}
\State $\begin{pmatrix}
x_1 \\
x_2\\
x_3
\end{pmatrix}
\longleftarrow
%\xleftarrow{\textbf{F}}
\begin{pmatrix}
2x_1+x_2^2+x_3\\
2x_2-x_2^2+2x_3\\
1-x_3
\end{pmatrix}
$

\EndWhile
}

\item \boxed{Fib1}\\[-0.14em]
\programboxappendix[0.55\linewidth]{
\State$(x_1, x_2,x_3)=(2,1,1)$
\While{true}
\State $\begin{pmatrix}
x_1 \\
x_2 \\
x_3
\end{pmatrix}
\longleftarrow
\begin{pmatrix}
x_2\\
x_3\\
2x_2x_3-x_1
\end{pmatrix}
$
\EndWhile
}

\item \boxed{Fib2}\\[-0.14em]
\programboxappendix[0.75\linewidth]{
\State$(x_1, x_2,x_3)=(3,-2,1)$
\While{true}
\State $\begin{pmatrix}
x_1 \\
x_2 \\
x_3
\end{pmatrix}
\longleftarrow
\begin{pmatrix}
x_2\\
2x_1x_3-x_2\\
4x_1x_2x_3-2x_1^2-2x_2^2+1
\end{pmatrix}
$
\EndWhile
}
\item \boxed{{Example  \, 9}}\\[-0.14em]
\programboxappendix[0.7\linewidth]{
\State$(x_1, x_2,x_3)=(3,-1,2)$
\While{true}
\State $\begin{pmatrix}
x_1 \\
x_2 \\
x_3
\end{pmatrix}
\longleftarrow
\begin{pmatrix}
x_1+x_3^2+1\\
x_2-x_3^2\\
x_3+(x_1+x_2)^2
\end{pmatrix}
$
\EndWhile
}
\item \boxed{Fib3}\\[-0.14em]
\programboxappendix[0.7\linewidth]{
\State$(x_1, x_2,x_3)=(2,-1,1)$
\While{true}
\State $\begin{pmatrix}
x_1 \\
x_2 \\
x_3
\end{pmatrix}
\longleftarrow
\begin{pmatrix}
1+x_1+x_2+x_1x_2-x_3\\
x_1\\
x_2
\end{pmatrix}
$
\EndWhile
}

\item \boxed{{Example  \, 10}}\\[-0.14em] 
\programboxappendix[0.8\linewidth]{
\State$(x_1, x_2,x_3)=(-1,2,1)$
\While{true}
\State $\begin{pmatrix}
x_1 \\
x_2 \\
x_3
\end{pmatrix}
\longleftarrow
\begin{pmatrix}
10x_1-8x_3+(x_1+x_2)^2\\
2x_2-(x_1+x_2)^2\\
6x_1-4x_3+(x_1+x_2)^2
\end{pmatrix}
$
\EndWhile
}
\newpage

\vspace{1cm}\item \boxed{Nagata}\\[-0.14em] 
\programboxappendix[0.91\linewidth]{
\State$(x_1, x_2, x_3)=(3,-2,5)$
\While{true}
\State $\begin{pmatrix}
x_1\\
x_2 \\
x_3
\end{pmatrix}
\longleftarrow
\begin{pmatrix}
x_1-2(x_1x_3+x_2^2)x_2-(x_1x_3+x_2^2)^2x_3\\
x_2+(x_1x_3+x_2^2)x_3\\
x_3
\end{pmatrix}
$
\EndWhile
}
\item \boxed{Yagzhev9}\\[-0.14em]
\programboxappendix[0.8\linewidth]{
\State$(x_1, \ldots,x_9)=(2,-3,1,4,-1,7,-4,3,2)$
\While{true}
\State $\begin{pmatrix}
x _1\\
x_2 \\
x_3\\
x_4\\
x_5\\
x_6\\
x_7\\
x_8\\
x_9
\end{pmatrix}
\longleftarrow
\begin{pmatrix}
x_1+x_1x_7x_9+x_2x_9^2\\
x_2-x_1x_7^2-x_2x_7x_9\\
x_3+x_3x_7x_9+x_4x_9^2\\
x_4-x_3x_7^2-x_4x_7x_9\\
x_5+x_5x_7x_9+x_6x_9^2\\
x_6-x_5x_7^2-x_6x_7x_9\\
x_7+(x_1x_4-x_2x_3)x_9\\
x_8+(x_3x_6-x_4x_5)x_9\\
\big(x_9+(x_1x_4-x_2x_3)x_8\\-(x_3x_6+x_4x_5)x_7\big)
\end{pmatrix}
$
\EndWhile
}
\item \boxed{Yagzhev11}\\[-0.14em]
\programboxappendix[0.81\linewidth]{
\State$(x_1,\ldots,x_{11})=(3,-1,2,1,-5,-1,3,4,-1,3,2)$
\While{true}
\State $\begin{pmatrix}
x _1\\
x_2 \\
x_3\\
x_4\\
x_5\\
x_6\\
x_7\\
x_8\\
x_9\\
x_{10}\\
x_{11}
\end{pmatrix}
\longleftarrow
\begin{pmatrix}
x_1-x_3x_{10}^2\\
x_2-x_3x_{11}^2\\
x_3+x_1x_{11}^2-x_2x_{10}^2\\
x_4-x_6x_{10}^2\\
x_5-x_6x_{11}^2\\
x_6+x_4x_{11}^2-x_5x_{10}^2\\
x_7-x_9x_{10}^2\\
x_8-x_9x_{11}^2 \\
x_9+x_7x_{11}^2-x_8x_{10}^2\\
x_{10}-\det(A)\\
x_{11}-x_{10}^3
\end{pmatrix}
$
\EndWhile

where $A = \begin{pmatrix}
x_1 & x_2 & x_3\\
x_4 & x_5 & x_6\\
x_7 & x_8 &x_9
\end{pmatrix}$
}
\item \boxed{Semi\text{-}algebraic}\\[-0.14em]
\programboxappendix[0.55\linewidth]{
\State$(x_1, x_2)=(a_1,a_2)$
\While{$x_1>0$}
\State $\begin{pmatrix}
x_1 \\
x_2
\end{pmatrix}
\longleftarrow
\begin{pmatrix}
2x_1\\
2x_2
\end{pmatrix}
$
\EndWhile
}

\vspace*{-1em}

%\newpage

%\item \boxed{Floor function}\\
 %   \programboxappendix[0.55\linewidth]{
%\State$(x_1, x_2, x_3)=(0,1,1)$
%\While{$x_2\leq N$}
%\State $\begin{pmatrix}
%x_1 \\
%x_2 \\
%x_3
%\end{pmatrix}
%\longleftarrow
%\begin{pmatrix}
%x_1+1\\
%x_2+x_3+2\\
%x_3+2
%\end{pmatrix}
%$
%\EndWhile
%}

\end{itemize}
 
\end{document}